\documentclass[a4paper,notitlepage,fleqn,12pt]{article}
\usepackage[tbtags]{amsmath}
\usepackage{amsfonts}
\usepackage{amsthm}
\usepackage{color}
\usepackage{amsfonts}
\usepackage{graphicx,amssymb}
\usepackage{natbib}
\numberwithin{equation}{section}
\newtheorem{thm}{Theorem}
\newtheorem{lemma}{Lemma}
\newtheorem{assum}{Assumption}

\DeclareMathOperator*{\argmin}{argmin}

\newcommand{\bm}[1]{\mbox{\boldmath{$#1$}}}

\renewcommand{\baselinestretch}{1.6}

\title{Quantile correlations and quantile autoregressive modeling}
\author{Guodong Li, Yang Li and Chih-Ling Tsai\\
\textit{University of Hong Kong and University of California at Davis}}

\begin{document}
\maketitle
\begin{abstract}
In this paper, we propose two important measures,  quantile correlation (QCOR)  and  quantile partial correlation (QPCOR). We then  apply them to quantile autoregressive (QAR) models, and introduce two valuable quantities, the quantile autocorrelation function (QACF) and the quantile partial  autocorrelation function (QPACF). This allows us to extend the classical Box-Jenkins approach to quantile autoregressive  models.
Specifically, the QPACF of an observed time series can be employed to identify the autoregressive order, while the QACF of residuals obtained from the fitted model can be used to assess the model adequacy. We not only demonstrate the asymptotic properties of QCOR, QPCOR, QACF, and
PQACF, but also show the large sample results of the QAR estimates and the quantile version of  the Ljung-Box test. Simulation studies indicate that the proposed methods perform well in finite samples, and an empirical example is presented to illustrate usefulness.
\end{abstract}

\newpage
\section{Introduction}

In the last decade, quantile regression has attracted considerable attention.
There are two major reasons for such popularity. The first is that
quantile regression estimation \citep{Koenker_Bassett1978} can be robust to non-Gaussian or heavy-tailed data. In addition, it includes the commonly used least absolute deviation (LAD) method as a special case. The second is that
 the quantile regression model allows practitioners to provide more easily interpretable regression
estimates obtained via various quantiles $\tau\in[0,1]$.
More references about quantile regression estimations and interpretations can be found in the seminal book of  \cite{Koenker2005}. Further extension of quantile regression to various model and data structures have been found in the literature,
 e.g., \cite{Machado_Silva2005} for count data, \cite{Mu_He2007} for power transformed data, \cite{Peng_Huang2008} and \cite{Wang_Wang2009} for survival analysis, \cite{He_Liang2000} and \cite{Wei_Carroll2009} for regression with measurement errors, \cite{Ando_Tsay2011} for regression with augmented factors, and \cite{Kai_Li_Zou2011} for semiparametric varying-coefficient partially linear models, among others.

In addition to the regression context,
the quantile technique has been employed to
 the field of time series; see, for example, \cite{Koul_Saleh1995} and \cite{Cai_Stander_Davies2012} for autoregressive (AR) models, \cite{Ling_McAleer2004} for unstable AR models, and \cite{Xiao_Koenker2009} for generalized autoregressive conditional heteroscedastic (GARCH) models.
It is noteworthy that
\cite{Koenker_Xiao2006} established important statistical properties for
quantile autoregressive (QAR) models, and suggested a modified Bayesian information criterion (BIC) to select the order of QAR models. Their findings have expanded the classical AR model into a new era,
which motivates us to extend the classical Box-Jenkins' approach (i.e., model identification, model  parameter estimation, and model diagnostics) from AR to QAR models. In the classical AR model, it is known that model identification usually relies on
the partial autocorrelation function (PACF) of the observed time series, while model diagnosis commonly depends on the autocorrelation function (ACF) of model residuals.
Detailed illustrations of model identification and diagnosis can be found in \cite{Box_Jenkins_Reinsel2008}.

The aim of this paper is to introduce two novel measures to
examine the  linear and partial linear relationships between any two random variables for
the given quantile $\tau\in[0,1]$. We name them quantile correlation (QCOR) and  quantile partial correlation (QPCOR). Based on these two measures, we
propose the quantile partial  autocorrelation function (QPACF) and the quantile autocorrelation function (QACF) to identify the order of the QAR model and to assess model adequacy, respectively.
It is noteworthy that the application of QCOR and QPCOR is not limited to QAR models. They can be used broadly as the classical correlation and partial correlation measures in various contexts.

The rest of this article is organized as follows. Section 2
introduces QCOR and QPCOR. Furthermore, the asymptotic properties of their sample estimators are established. Section 3
obtains QPACF and its large sample property for identifying the order of QAR model.
In addition, the autoregressive parameter estimator and its asymptotic distribution are demonstrated.
Moreover, QACF and its resulting test statistics, together with their asymptotic results, are provided to examine the model adequacy.
Section 4 conducts simulation experiments to study the finite sample performance of the proposed methods, and also presents an empirical example to demonstrate  usefulness. Finally,  we  conclude the article with a brief discussion in Section 5.  All technical proofs of lemmas and theorems are relegated to the Appendix.

\section{Correlations}
\subsection{Quantile correlation and quantile partial correlation}

For random variables $X$ and $Y$, let $Q_{\tau,Y}$ be the $\tau$th unconditional quantile of $Y$
and $Q_{\tau,Y}(X)$ be the $\tau$th quantile of $Y$ conditional on $X$.
One can show that $Q_{\tau,Y}(X)$ is independent of $X$, i.e. $Q_{\tau,Y}(X)=Q_{\tau,Y}$ with probability one, if and only if the random variables $I(Y-Q_{\tau,Y}>0)$ and $X$ are independent, where $I(\cdot)$ is the indicated function. This fact has been
used by  \cite{He_Zhu2003} and  \cite{Mu_He2007}, and it also motivates us to
define the quantile covariance given below.
For $0<\tau<1$, define
\[
\textrm{qcov}_{\tau}\{Y,X\}=\textrm{cov}\{I(Y-Q_{\tau,Y}>0),X\}=E\{\psi_{\tau}(Y-Q_{\tau,Y})(X-EX)\},
\]
where the function $\psi_\tau(w)=\tau-I(w<0)$.
Subsequently, the quantile correlation can be defined as follows,
\begin{equation}\label{eqc}
\textrm{qcor}_{\tau}\{Y,X\}=\frac{\textrm{qcov}_{\tau}\{Y,X\}}{\sqrt{\textrm{var}\{\psi_{\tau}(Y-Q_{\tau,Y})\}\textrm{var}(X)}} =\frac{E\{\psi_{\tau}(Y-Q_{\tau,Y})(X-EX)\}}{\sqrt{(\tau-\tau^2)\sigma_X^2}},
\end{equation}
where $\sigma_X^2=\textrm{var}(X)$.

In the simple linear regression with the quadratic loss function, there is
a nice relationship between the slope and correlation. Hence, it is of interest to
find a connection between the quantile slope and $\textrm{qcov}_{\tau}\{Y,X\}$. To this end,
 consider a simple quantile linear regression,
\[
(a_0,b_0)=\argmin_{a,b}E[\rho_{\tau}(Y-a-bX)],
\]
in which one attempts to approximate $Q_{\tau,Y}(X)$ by a linear function $a_0+b_0X$ (see \citealp{Koenker2005}), where $\rho_\tau(w)=w[\tau-I(w<0)]$.
Then, we obtain the relationship between $b_0$ and $\textrm{qcor}_{\tau}\{Y,X\}$ given below.

\begin{lemma}\label{lem1}
Suppose that random variables $X$ and $Y$ have a joint density and $EX^2<\infty$. Then the values of $(a_0,b_0)$ are unique, and the quantity $b_0=0$ if and only if the quantile correlation $\textrm{qcor}_{\tau}\{Y,X\}=0$.
\end{lemma}
\noindent
It is noteworthy that the proposed quantile covariance here does not enjoy the symmetry property of the classical covariance, i.e., $\textrm{qcov}_{\tau}(Y,X)\neq\textrm{qcov}_{\tau}(X,Y)$. This is because
the first argument of the quantile covariance or the quantile correlation is related to the $\tau$th quantile, while the second argument is the same as that of the classical covariance.
Accordingly,  $\textrm{qcor}_{\tau}(Y,X)\neq\textrm{qcor}_{\tau}(X,Y)$.

Suppose that a quantile linear regression model has the response $Y$, a
$q\times 1$ vector of covariate $\mathbf{Z}$,
and  an  additional covariate $X$. In the classical regression model, one can construct the partial correlation to measure the linear relationship between variables $Y$ and $X$ after adjusting (or controlling) vector $\mathbf{Z}$ (e.g., see \citealp{Chatterjee_Hadi2006}). This motivates us to
propose the quantile partial correlation function. To this end,
let
\[
(\alpha_1,\beta_1^{\prime})=\argmin_{\alpha,\beta}E(X-\alpha-\beta^{\prime}\mathbf{Z})^2,
\]
where $(\alpha,\beta^{\prime})^{\prime}$ is a vector of unknown parameters.
Accordingly, $\alpha_1+\beta_1^{\prime}\mathbf{Z}$ is the linear effect of $\mathbf{Z}$ on $X$.
Next, consider
\[
(\alpha_2,\beta_2^{\prime})=\argmin_{\alpha,\beta}E[\rho_{\tau}(Y-\alpha-\beta^{\prime}\mathbf{Z})].
\]
As a result, $\alpha_2+\beta_2^{\prime}\mathbf{Z}$ is the linear effect of $\mathbf{Z}$ on the quantile $Y$ (i.e., the linear approximation of $Q_{\tau,Y}(\mathbf{Z})$).
It can also  be shown that $E(X-\alpha_1-\beta_1^{\prime}\mathbf{Z})=0$, $E[\psi_{\tau}(Y-\alpha_2-\beta_3^{\prime}\mathbf{Z})]=0$ and $E[\psi_{\tau}(Y-\alpha_2-\beta_3^{\prime}\mathbf{Z})\mathbf{Z}]=0$ if the random vector $(X,Y,\mathbf{Z}^{\prime})^{\prime}$ satisfies the conditions stated in the forthcoming Lemma \ref{lem2}. Using these facts, we define
 the quantile partial correlation as follows,
\begin{align}
\begin{split}\label{eqd}
\textrm{qpcor}_{\tau}\{Y,X|\mathbf{Z}\} &=\frac{\textrm{cov}\{\psi_{\tau}(Y-\alpha_2-\beta_2^{\prime}\mathbf{Z}),X-\alpha_1-\beta_1^{\prime}\mathbf{Z}\}} {\sqrt{\textrm{var}\{\psi_{\tau}(Y-\alpha_2-\beta_2^{\prime}\mathbf{Z})\} \textrm{var}\{X-\alpha_1-\beta_1^{\prime}\mathbf{Z}\}}}\\
&=\frac{E[\psi_{\tau}(Y-\alpha_2-\beta_2^{\prime}\mathbf{Z})(X-\alpha_1-\beta_1^{\prime}\mathbf{Z})]} {\sqrt{(\tau-\tau^2)E(X-\alpha_1-\beta_1^{\prime}\mathbf{Z})^2}}\\
&=\frac{E[\psi_{\tau}(Y-\alpha_2-\beta_2^{\prime}\mathbf{Z})X]} {\sqrt{(\tau-\tau^2)\sigma_{X|\mathbf{Z}}^2}},
\end{split}
\end{align}
where $\sigma_{X|\mathbf{Z}}^2=E(X-\alpha_1-\beta_1^{\prime}\mathbf{Z})^2$.
This indicates that the covariate $X$ has no additional linear contribution to the quantile
response $Y$ if $\alpha_2+\beta_2^{\prime}\mathbf{Z}=\alpha_3+\beta_3^{\prime}\mathbf{Z}+\gamma_3X$ with probability one, where
\[
(\alpha_3,\beta_3^{\prime},\gamma_3)=\argmin_{\alpha,\beta,\gamma} E[\rho_{\tau}(Y-\alpha-\beta^{\prime}\mathbf{Z}-\gamma X)].
\]
This leads to the following lemma.
\begin{lemma}\label{lem2}
Suppose that the random vector $(X,Y,\mathbf{Z}^{\prime})^{\prime}$ has a joint density with $EX^2<\infty$ and $E\|\mathbf{Z}\|^2<\infty$, where $\|\cdot\|$ is the Euclid norm. Then $(\alpha_3,\beta_3^{\prime},\gamma_3)=(\alpha_2,\beta_2^{\prime},0)$ if and only if the quantile partial correlation $\textrm{qpcor}_{\tau}\{Y,X|\mathbf{Z}\}=0$.
\end{lemma}
\noindent
Since the true $\textrm{qcor}_{\tau}$ and $\textrm{qpcor}_{\tau}$  are often unknown in practice, we introduce their sample versions given below.

\subsection{Sample quantile correlation and sample quantile partial correlation}

Suppose that the data $\{(Y_i,X_i,\mathbf{Z}_i^{\prime})^{\prime},i=1,...,n\}$ are  identically and independently  generated from a distribution of
$(Y,X,\mathbf{Z}^{\prime})^{\prime}$. Let $\widehat{Q}_{\tau,Y}=\inf\{y: F_n(y)\geq \tau\}$ be the sample $\tau$th quantile of $Y_1,...,Y_n$, where $F_n(y)=n^{-1}\sum_{i=1}^nI(Y_i\leq y)$ is the empirical distribution function. Based on equation (\ref{eqc}), the sample estimate of the quantile correlation $\textrm{qcor}_{\tau}\{Y,X\}$ is defined as

\begin{equation}\label{seqc}
\widehat{\textrm{qcor}}_{\tau}\{Y,X\} =\frac{1}{\sqrt{(\tau-\tau^2)\widehat{\sigma}_X^2}}\cdot \frac{1}{n}\sum_{i=1}^n \psi_{\tau}(Y_i-\widehat{Q}_{\tau,Y})(X_i-\bar{X}),
\end{equation}
where $\bar{X}=n^{-1}\sum_{i=1}^nX_i$, and $\widehat{\sigma}_X^2=n^{-1}\sum_{i=1}^n(X_i-\bar{X})^2$.

To study the asymptotic property of
$\widehat{\textrm{qcor}}_{\tau}\{Y,X\}$,
denote $f_Y(\cdot)$ and $f_{Y|X}(\cdot)$ as the density of $Y$ and the conditional density of $Y$ given $X$,
respectively. In addition,
let $\mu_X=E(X)$, $\mu_{X|Y}=E[f_{Y|X}({Q}_{\tau,Y})X]/f_Y(Q_{\tau,Y})$, $\Sigma_{11}=E(X-\mu_X)^4-\sigma_X^4$, $$\Sigma_{12}=E[\psi_{\tau}(Y-{Q}_{\tau,Y})(X-\mu_{X|Y})]^2-[\textrm{qcov}_{\tau}\{Y,X\}]^2,$$ $$\Sigma_{13}=E[\psi_{\tau}(Y-{Q}_{\tau,Y})(X-\mu_{X|Y})(X-\mu_X)^2]-\sigma_X^2\cdot\textrm{qcov}_{\tau}\{Y,X\},
$$ and
\[
\Omega_1=\frac{1}{\tau-\tau^2}\left[
\frac{\Sigma_{11}(\textrm{qcov}_{\tau}\{Y,X\})^2}{4\sigma_X^6} -\frac{\Sigma_{13}\cdot\textrm{qcov}_{\tau}\{Y,X\}}{\sigma_X^4} +\frac{\Sigma_{12}}{\sigma_X^2}
\right],
\]
where $\sigma_X^2$ is defined as in the previous subsection.
Then, we obtain the following result.

\begin{thm}\label{thm1}
Suppose that $E(X)^4<\infty$ and there exists a $\pi>0$ such that the density $f_Y(\cdot)$ is continuous and the conditional density $f_{Y|X}(\cdot)$ is uniformly integrable on $[{Q}_{\tau,Y}-\pi,{Q}_{\tau,Y}+\pi]$. Then
\[
\sqrt{n}\left(\widehat{\textrm{qcorr}}_{\tau}\{Y,X\}-\textrm{qcorr}_{\tau}\{Y,X\}\right)\rightarrow_dN(0,\Omega_1).
\]
\end{thm}

To apply the above theorem, one needs to estimate the asymptotic variance $\Omega_1$.  To this end, we
employ a nonparametric approach, such as the Nadaraya-Watson regression, to estimate the function
$m(y)=E(X|Y=y)$, and denote it as $\widehat{m}(y)$. We further assume that the random vector $(X,Y)$ has a joint density, which leads to  $\mu_{X|Y}=E(X|Y=Q_{\tau,Y})$. Accordingly, we obtain the estimate,
$\widehat{\mu}_{X|Y}=\widehat{m}(\widehat{Q}_{\tau,Y})$, where $\widehat{Q}_{\tau,Y}$ is the
$\tau$th sample  quantile of $\{Y_1,...,Y_n\}$.
Finally, the rest of
quantities,   $\mu_X$, $\sigma_X^2$, $\text{qcov}_{\tau}\{Y,X\}$, $\Sigma_{11}$, $\Sigma_{12}$, and $\Sigma_{13}$ contained in $\Omega_1$ can be,  respectively, estimated by
$\widehat{\mu}_X=\bar{X}=n^{-1}\sum_{i=1}^nX_i$, $\widehat{\sigma}_X^2=n^{-1}\sum_{i=1}^n(X_i-\widehat{\mu}_X)^2$,
$\widehat{\textrm{qcov}}_{\tau}\{Y,X\} =n^{-1}\sum_{i=1}^n \psi_{\tau}(Y_i-\widehat{Q}_{\tau,Y})(X_i-\bar{X})$, $\widehat{\Sigma}_{11}=n^{-1}\sum_{i=1}^n(X_i-\widehat{\mu}_X)^4 -\widehat{\sigma}_X^4$, $\widehat{\Sigma}_{12}=n^{-1}\sum_{i=1}^n[\psi_{\tau}(Y_i-\widehat{Q}_{\tau,Y})(X_i-\widehat{\mu}_{X|Y})]^2 -[\widehat{\textrm{qcov}}_{\tau}\{Y,X\}]^2$, and $\widehat{\Sigma}_{13}=n^{-1}\sum_{i=1}^n\psi_{\tau}(Y_i-\widehat{Q}_{\tau,Y})(X_i-\widehat{\mu}_{X|Y})(X_i-\widehat{\mu}_X)^2 -\widehat{\sigma}_X^2\widehat{\textrm{qcov}}_{\tau}\{Y,X\}$. As a result, we obtain the estimate of
$\Omega_1$, and denote it by $\widehat{\Omega}_1$.

We next estimate the quantile partial correlation $\textrm{qpcor}_{\tau}\{Y,X\}$.
Let
\[
(\widehat{\alpha}_1,\widehat{\beta}_1^{\prime})=\argmin_{\alpha,\beta}\sum_{i=1}^n(X_i-\alpha-\beta^{\prime}\mathbf{Z}_i)^2
\hspace{5mm}
\text{and}
\hspace{5mm}
(\widehat{\alpha}_2,\widehat{\beta}_2^{\prime}) =\argmin_{\alpha,\beta}\sum_{i=1}^n\rho_{\tau}(Y_i-\alpha-\beta^{\prime}\mathbf{Z}_i).
\]
Based on equation (\ref{eqd}), the sample quantile partial correlation is defined as
\begin{equation}\label{seqd}
\widehat{\textrm{qpcor}}_{\tau}\{Y,X|\mathbf{Z}\}=\frac{1}{\sqrt{(\tau-\tau^2)\widehat{\sigma}_{X|\mathbf{Z}}^2}} \cdot \frac{1}{n}\sum_{i=1}^n \psi_{\tau}(Y_i-\widehat{\alpha}_2-\widehat{\beta}_2^{\prime}\mathbf{Z}_i) X_i,
\end{equation}
where $\widehat{\sigma}_{X|\mathbf{Z}}^2=n^{-1}\sum_{i=1}^n(X_i-\widehat{\alpha}_1-\widehat{\beta}_1^{\prime}\mathbf{Z}_i)^2$.

To investigate the asymptotic property of $\widehat{\textrm{qpcor}}_{\tau}\{Y,X|\mathbf{Z}\}$,
denote the conditional density of $Y$ given $\mathbf{Z}$ and  the conditional density of
$Y$ given $\mathbf{Z}$ and $X$  by $f_{Y|\mathbf{Z}}(\cdot)$  and $f_{Y|\mathbf{Z},X}(\cdot)$, respectively. In addition,
let $\theta_1=(\alpha_1,\beta_1^{\prime})^{\prime}$, $\theta_2=(\alpha_2,\beta_2^{\prime})^{\prime}$, $\mathbf{Z}^*=(1,\mathbf{Z}^{\prime})^{\prime}$,
$\Sigma_{21}=E[f_{Y|\mathbf{Z},X}(\theta_2^{\prime}\mathbf{Z}^*)X\mathbf{Z}^*]$,
$\Sigma_{22}=E[f_{Y|\mathbf{Z}}(\theta_2^{\prime}\mathbf{Z}^*)\mathbf{Z}^*\mathbf{Z}^{*\prime}]$,
$\Sigma_{20}=\Sigma_{21}^{\prime} \Sigma_{22}^{-1}$,
$\Sigma_{23}=E(X-\theta_1^{\prime}\mathbf{Z}^*)^4-\sigma_{X|\mathbf{Z}}^4$,
\[
\Sigma_{24}=E[\psi_{\tau}(Y-\theta_2\mathbf{Z}^*)(X -\Sigma_{20}\mathbf{Z}^*)]^2-\{E[\psi_{\tau}(Y-\theta_2^{\prime}\mathbf{Z}^*)X]\}^2,
\]
\[
\Sigma_{25}=E[\psi_{\tau}(Y-\theta_2\mathbf{Z}^*)(X -\Sigma_{20}\mathbf{Z}^*)(X-\theta_1^{\prime}\mathbf{Z}^*)^2] -\sigma_{X|\mathbf{Z}}^2\cdot E[\psi_{\tau}(Y-\theta_2^{\prime}\mathbf{Z}^*)X],
\]
and
\[
\Omega_2=\frac{1}{\tau-\tau^2}\left[
\frac{\Sigma_{23}(E[\psi_{\tau}(Y-\theta_2^{\prime}\mathbf{Z}^*)X])^2}{4\sigma_{X|\mathbf{Z}}^6} -\frac{\Sigma_{25}\cdot E[\psi_{\tau}(Y-\theta_2^{\prime}\mathbf{Z}^*)X]}{\sigma_{X|\mathbf{Z}}^4} +\frac{\Sigma_{24}}{\sigma_{X|\mathbf{Z}}^2}
\right],
\]
where $\alpha_1$, $\beta_1$, $\alpha_2$, $\beta_2$ and $\sigma_{X|\mathbf{Z}}^2$ are defined as in the previous subsection. Then, we have the following result.

\begin{thm}\label{thm2}
Suppose that $\Sigma_{21}<\infty$, $0<\Sigma_{22}<\infty$, $EX^4<\infty$, $E\|\mathbf{Z}\|^4<\infty$, $E(\mathbf{Z}^*\mathbf{Z}^{*\prime})>0$, and there exists a $\pi>0$ such that $f_{Y|\mathbf{Z}}(\theta_2^{\prime}\mathbf{Z}^*+\cdot)$ and $f_{Y|\mathbf{Z},X}(\theta_2^{\prime}\mathbf{Z}^*+\cdot)$ are uniformly integrable on $[-\pi,\pi]$.
Then
\[
\sqrt{n}[\widehat{\textrm{qpcor}}_{\tau}\{Y,X|\mathbf{Z}\} -\textrm{qpcor}_{\tau}\{Y,X|\mathbf{Z}\}]\rightarrow_dN(0,\Omega_2).
\]
\end{thm}

To estimate the asymptotic variance $\Omega_2$ given in  Theorem 2, we consider
$Y^*=Y-\theta_2^{\prime}\mathbf{Z}^*$ and  $\textrm{qcov}_{\tau}\{Y^*,X\}=E[\psi_{\tau}(Y-\theta_2^{\prime}\mathbf{Z}^*)X]$. In addition,
assume that the random vector $(Y,X,\mathbf{Z}^{\prime})^{\prime}$ has a joint density.
We then have that $\Sigma_{21}= E[f_{Y^*|\mathbf{Z},X}(0)X\mathbf{Z}^*]=f_{Y^*}(0)\cdot E[X\mathbf{Z}^*|Y^*=0]$, $\Sigma_{22}=f_{Y^*}(0)\cdot E[\mathbf{Z}^*\mathbf{Z}^{*\prime}|Y^*=0]$, and $\Sigma_{20}=E[X\mathbf{Z}^{*\prime}|Y^*=0]\{E[\mathbf{Z}^*\mathbf{Z}^{*\prime}|Y^*=0]\}^{-1}$, where $f_{Y^*}(\cdot)$ is the density of $Y^*$.
Applying the same nonparametric technique as that used for estimating $\mu_{X|Y}$ in Theorem 1,
we could estimate each of the vector and matrix components in $\mathbf{m}_1(y)=E[X\mathbf{Z}^*|Y^*=y]$
and $\mathbf{m}_2(y)=E[\mathbf{Z}^*\mathbf{Z}^{*\prime}|Y^*=y]$, respectively, from the data
$\{(Y_i^*,X_i,\mathbf{Z}_i^{\prime})=(Y_i-\widehat{\theta}_2^{\prime}\mathbf{Z}^*_i,X_i,\mathbf{Z}_i^{\prime}),i=1,...,n\}$, where $\widehat{\theta}_2=(\widehat\alpha_2,\widehat\beta_2^{\prime})^{\prime}$.
Accordingly, $\widehat{\Sigma}_{20}=\widehat\Sigma_{21}^{\prime} \widehat\Sigma_{22}^{-1}=\widehat{\mathbf{m}}_1^{\prime}(0)[\widehat{\mathbf{m}}_2(0)]^{-1}$.
Subsequently, the
rest of quantities involved in  $\Omega_2$,  $\sigma_{X|\mathbf{Z}}^2$, $\textrm{qcov}_{\tau}\{Y^*,X\}$, $\Sigma_{23}$, $\Sigma_{24}$, and $\Sigma_{25}$ can be, respectively, estimated by
$\widehat{\sigma}_{X|\mathbf{Z}}^2=n^{-1}\sum_{i=1}^n(X_i-\widehat{\alpha}_1-\widehat{\beta}_1^{\prime}\mathbf{Z}_i)^2$, $\widehat{\textrm{qcov}}_{\tau}\{Y^*,X\}=n^{-1}\sum_{i=1}^n\psi_{\tau}(Y_i-\widehat{\theta}_2^{\prime}\mathbf{Z}^*_i)X_i$, $\widehat{\Sigma}_{23}=n^{-1}\sum_{i=1}^n(X_i-\widehat{\theta}_1^{\prime}\mathbf{Z}^*_i)^4-\widehat{\sigma}_{X|\mathbf{Z}}^4$, $\widehat{\Sigma}_{24}=n^{-1}\sum_{i=1}^n[\psi_{\tau}(Y_i-\widehat{\theta}_2\mathbf{Z}^*_i)(X_i -\widehat{\Sigma}_{20}\mathbf{Z}^*_i)]^2-[\widehat{\textrm{qcov}}_{\tau}\{Y^*,X\}]^2$, and $\widehat{\Sigma}_{25}=n^{-1}\sum_{i=1}^n\psi_{\tau}(Y_i-\widehat{\theta}_2\mathbf{Z}^*_i)(X_i -\widehat{\Sigma}_{20}\mathbf{Z}^*_i)(X_i-\widehat{\theta}_1^{\prime}\mathbf{Z}^*_i)^2 -\widehat{\sigma}_{X|\mathbf{Z}}^2\cdot\widehat{\textrm{qcov}}_{\tau}\{Y^*,X\}$. Consequently, we obtain the estimate of
$\Omega_2$, and denote it by $\widehat{\Omega}_2$.

It is noteworthy that the quantile correlation and  quantile partial correlation
can be broadly used as the classical correlation and  partial correlation in regression analysis (e.g., variable selections), although our focus is on quantile autoregressive models.

\section{Quantile autoregressive modeling}

Suppose that $\{y_t\}$ is a strictly stationary and ergodic time series, and $\mathcal{F}_t$ is the $\sigma$-field generated by $\{y_t,y_{t-1},...\}$. We then follow
\citeauthor{Koenker_Xiao2006}'s (2006) approach and present QAR models; i.e.,
conditional on $\mathcal{F}_{t-1}$, the $\tau$th quantile of $y_t$  has the form of
\begin{equation}\label{section3_eq1}
Q_{\tau}(y_t|\mathcal{F}_{t-1})=\phi_0(\tau)+\phi_1(\tau)y_{t-1}+\cdots +\phi_p(\tau)y_{t-p}\hspace{2mm} \text{for}\hspace{2mm} 0<\tau<1,
\end{equation}
where  $\phi_i(\cdot)$s are unknown functions mapping from $[0,1]\rightarrow R$.
Following the Box-Jenkins' classical approach, we next introduce the QPACF of a time series to identify the order of a QAR model, and then propose using the QACF of residuals to assess the adequacy of the fitted model.

\subsection{Model identification and estimation}

For the positive integer $k$, let $\mathbf{z}_{t,k-1}=(y_{t-1},...,y_{t-k+1})^{\prime}$,
$(\alpha_1,\beta_1^{\prime}) =\argmin_{\alpha,\beta} E(y_{t-k}-\alpha-\beta^{\prime}\mathbf{z}_{t,k-1})^2$,
and $(\alpha_2,\beta_2^{\prime}) =\argmin_{\alpha,\beta}E[\rho_{\tau}(y_t-\alpha-\beta^{\prime}\mathbf{z}_{t,k-1})]$,
where  the notations $(\alpha_1,\beta_1^{\prime})$ and $(\alpha_2,\beta_2^{\prime})$ are a slight abuse since they have been used to denote the regression parameters in Section 2.
From  equation (\ref{eqd}), we obtain
 the quantile partial correlation between $y_t$ and $y_{t-k}$ after adjusting the
linear effect  $\mathbf{z}_{t,k-1}$,
\begin{equation*}
\phi_{kk,\tau}=\textrm{qpcor}_{\tau}\{y_t,y_{t-k}|\mathbf{z}_{t,k-1}\} =\frac{E[\psi_{\tau}(y_t-\alpha_2-\beta_2^{\prime}\mathbf{z}_{t,k-1})y_{t-k}]}{\sqrt{(\tau-\tau^2)E(y_{t-k}-\alpha_1-\beta_1^{\prime}\mathbf{z}_{t,k-1})^2}},
\end{equation*}
and it is independent of the time index $t$ due to the strict stationarity of $\{y_t\}$. Analogous to the definition of the classical PACF \cite[Chapter 2]{Fan_Yao2003}, we name $\phi_{kk,\tau}$ to be the QPACF of time series $\{y_t\}$. It is also noteworthy that $\phi_{11,\tau}=\textrm{qcor}_{\tau}\{y_t,y_{t-1}\}$.
We next show the cut-off property of QPACF.

\begin{lemma}\label{lem3}
If $\phi_p(\tau)\neq 0$ with $p>0$, $Ey_t^2<\infty$ and $E[y_t-E(y_t|\mathcal{F}_{t-1})]^2>0$, then $\phi_{pp,\tau}\neq 0$, and $\phi_{kk,\tau}=0$ for $k>p$.
\end{lemma}
\noindent
The above lemma indicates that the proposed QPACF plays the same role as that of PACF in the classical AR model identification.

In practice, one needs the sample estimate of QPACF. To this end,
let
\[
(\widetilde{\alpha}_1,\widetilde{\beta}_1^{\prime}) =\argmin_{\alpha,\beta}\sum_{t=k+1}^n(y_{t-k}-\alpha-\beta^{\prime}\mathbf{z}_{t,k-1})^2,\hspace{5mm}
(\widetilde{\alpha}_2,\widetilde{\beta}_2^{\prime}) =\argmin_{\alpha,\beta}\sum_{t=k+1}^n\rho_{\tau}(y_t-\alpha-\beta^{\prime}\mathbf{z}_{t,k-1}),
\]
and $\widetilde{\sigma}_{y|\mathbf{z}}^2=n^{-1} \sum_{t=k+1}^n(y_{t-k}-\widetilde{\alpha}_1-\widetilde{\beta}_1^{\prime}\mathbf{z}_{t,k-1})^2$.
According to (\ref{seqd}), we obtain the estimation for $\phi_{kk,\tau}$,
\[
\widetilde{\phi}_{kk,\tau}=\frac{1}{\sqrt{(\tau-\tau^2)\widetilde{\sigma}_{y|\mathbf{z}}^2}}\cdot \frac{1}{n}\sum_{t=k+1}^n \psi_{\tau}(y_t-\widetilde{\alpha}_2-\widetilde{\beta}_2^{\prime}\mathbf{z}_{t,k-1}) y_{t-k},
\]
and we term it the sample QPACF of the time series.

To study the asymptotic property of $\widetilde{\phi}_{kk,\tau}$, we
introduce the following assumption, which is similar to Condition A.3 in \cite{Koenker_Xiao2006}.
\begin{assum}\label{assum3}
$Ey_t^2<\infty$, $E[y_t-E(y_t|\mathcal{F}_{t-1})]^2>0$, and there exists a $\pi>0$ such that $f_{t-1}(\cdot)$ is uniformly integrable on $[-\pi,\pi]$.
\end{assum}
\noindent
Furthermore,
let
\begin{equation}\label{eqb}
e_{t,\tau}=y_t-\phi_0(\tau)-\phi_1(\tau)y_{t-1}-\cdots -\phi_p(\tau)y_{t-p}.
\end{equation}
By \eqref{section3_eq1}, the random variable $I(e_{t,\tau}>0)$ is independent of $y_{t-k}$ for any $k>0$, and
$(\alpha_2,\beta_2^{\prime})=(\phi_0(\tau),\phi_1(\tau),...,\phi_p(\tau),0,...,0) $
for $k> p$.
Let $f_{t-1}(\cdot)$ be the conditional density of $e_{t,\tau}$ on the $\sigma$-field $\mathcal{F}_{t-1}$, and $\mathbf{z}^*_{t,k-1}=(1,\mathbf{z}_{t,k-1}^{\prime})^{\prime}=(1,y_{t-1},...,y_{t-k+1})^{\prime}$.
Moreover, let $A_0=E[y_{t-k}\mathbf{z}^*_{t,k-1}]$, $A_1=E[f_{t-1}(0)y_{t-k}\mathbf{z}^*_{t,k-1}]$, $\Sigma_{30}=E[\mathbf{z}^*_{t,k-1}\mathbf{z}_{t,k-1}^{*\prime}]$,  $\Sigma_{31}=E[f_{t-1}(0)\mathbf{z}^*_{t,k-1}\mathbf{z}_{t,k-1}^{*\prime}]$, and
\[
\Omega_3=\frac{E(y_t^2)-2A_1^{\prime}\Sigma_{31}^{-1}A_0 +A_1^{\prime}\Sigma_{31}^{-1}\Sigma_{30}\Sigma_{31}^{-1}A_1}{E(y_{t-k}-\alpha_1-\beta_1^{\prime}\mathbf{z}_{t,k-1})^2}.
\]
Then, we obtain the asymptotic result given below.
\begin{thm}\label{thm3}
For $k>p$, if $A_1<\infty$, $0<\Sigma_{31}<\infty$ and Assumption \ref{assum3} is satisfied, then ${\phi}_{kk,\tau}=0$ and
\[
\sqrt{n}\widetilde{\phi}_{kk,\tau}\rightarrow_dN(0,\Omega_3).
\]
\end{thm}
\noindent
To estimate  $\Omega_3$ in the above theorem, we first apply the
\cite{Hendricks_Koenker1991} method to obtain the estimation of $f_{t-1}(0)$ given below.
\[
\widetilde{f}_{t-1}(0) =\frac{2h}{\widetilde{Q}_{\tau+h}(y_t|\mathcal{F}_{t-1})-\widetilde{Q}_{\tau-h}(y_t|\mathcal{F}_{t-1})},
\]
where $\widetilde{Q}_{\tau}(y_t|\mathcal{F}_{t-1})=\widetilde\phi_0(\tau)+\widetilde\phi_1(\tau)y_{t-1}+\cdots +\widetilde\phi_k(\tau)y_{t-k}$ is the estimated $\tau$th quantile of $y_t$ and
$h$ is the bandwidth selected via appropriate methods (e.g., see \citealp{Koenker_Xiao2006}).
Afterwards, we can use the sample averaging to approximate $A_0$, $A_1$, $\Sigma_{30}$, $\Sigma_{31}$, $E(y_t^2)$, and $E(y_{t-k}-\alpha_1-\beta_1^{\prime}\mathbf{z}_{t,k-1})^2$ by replacing their $f_{t-1}(\cdot)$, $\alpha_1$, and $\beta_1$, respectively, with $\widetilde{f}_{t-1}(0)$, $\widetilde{\alpha}_1$ and $\widetilde{\beta}_1$. Accordingly, we obtain an estimate of $\Omega_3$, and denote it as $\widehat\Omega_3$.
In sum, we are able to use the threshold values $\pm 1.96\sqrt{\widehat\Omega_3/n}$ to check the significance of $\widetilde{\phi}_{kk,\tau}$.

To demonstrate how to use the above theorem to identify the order of a QAR model,  we generate the observations $y_1,...,y_{200}$ from
$
y_t=\Phi^{-1}(u_t)+a(u_t)y_{t-1},
$
where $\Phi$ is the standard normal cumulative distribution function, $a(x)=\max\{0.8-1.6x,0\}$, and $\{u_t\}$ is an $i.i.d$ sequence with uniform distribution on $[0,1]$. We attempt to fit
the QAR model \eqref{section3_eq1} with $\tau=0.2$, 0.4, 0.6,  and 0.8, respectively, to the observed data $\{y_t\}$.
Figure \ref{fig1} presents the sample QPACF $\widetilde{\phi}_{kk,\tau}$ for each $\tau$ with
the reference lines
$\pm 1.96\sqrt{\widehat\Omega_3/n}$.
We may conclude that the order $p$ is 1
when $\tau=0.2$ and  0.4, while  $p$ is 0 when $\tau=0.6$ and 0.8.

After the order $p$ of model \eqref{section3_eq1} is correctly identified, we subsequently fit
the selected model to data.
Let  $\bm{\phi}=(\phi_0,\phi_1,...,\phi_p)^{\prime}$ be an any parameter vector in  model \eqref{section3_eq1} and $\bm{\phi}(\tau)=((\phi_0(\tau),\phi_1(\tau),...,\phi_p(\tau))^{\prime}$ be the true value of $\bm{\phi}$.
Consider
\[
\widetilde{\bm{\phi}}(\tau)=\argmin_{\bm{\phi}}\sum_{t=p+1}^n\rho_{\tau}(y_t-\bm{\phi}^{\prime}\mathbf{z}^*_{t,p}),
\]
where $\mathbf{z}^*_{t,p}=(1,\mathbf{z}_{t,p}^{\prime})^{\prime}=(1,y_{t-1},...,y_{t-p})^{\prime}$.
In addition, let $\Sigma_{40}=E[\mathbf{z}^*_{t,p}\mathbf{z}_{t,p}^{*\prime}]$, $\Sigma_{41}=E[f_{t-1}(0)\mathbf{z}^*_{t,p}\mathbf{z}_{t,p}^{*\prime}]$, and $\Omega_4=(\tau-\tau^2)\Sigma_{41}^{-1}\Sigma_{40}\Sigma_{41}^{-1}$. We then obtain the following asymptotic property of the estimated parameter.

\begin{thm}\label{thm4}
If $0<\Sigma_{41}<\infty$ and Assumption \ref{assum3} is satisfied, then
\[
\sqrt{n}\{\widetilde{\bm{\phi}}(\tau)-\bm{\phi}(\tau)\}\rightarrow_d N(0,\Omega_4).
\]
\end{thm}
\noindent
The above result is similar to that of Theorem 2 in \cite{Koenker_Xiao2006}, although we make
different assumptions. The $\Omega_4$ in the above theorem can be estimated by applying the same techniques used for the estimation of $\Omega_3$.

\subsection{Model  diagnostics}

For the errors $\{e_{t,\tau}\}$ defined  in \eqref{eqb}, we employ
equation (\ref{eqc}) and the fact that ${Q}_{\tau,e_{t,\tau}}=0$, and  obtain
 QACF between $\{e_{t,\tau}\}$ and $\{e_{t-k,\tau}\}$ as follows,
\[
\rho_{k,\tau}=\frac{E\{\psi_{\tau}(e_{t,\tau})[e_{t-k,\tau}-E(e_{t,\tau})]\}} {\sqrt{(\tau-\tau^2)
{\sigma}_e^2}},
\]
where ${\sigma}_e^2=\text{var}(e_{t,\tau})$.
Suppose that the QAR model is correctly specified.
We can show that $\rho_{k,\tau}=0$ for $k>0$. Hence, we are able to use
$\rho_{k,\tau}$ to assess the model fit.
In the sample version, we consider the residuals of the QAR model,
\[
\widetilde{e}_{t,\tau}=y_t-\widetilde{\phi}_0(\tau)-\widetilde{\phi}_1(\tau)y_{t-1}-\cdots-\widetilde{\phi}_p(\tau)y_{t-p},
\]
for $t=p+1,...,n$, and $\widetilde{e}_{t,\tau}=0$ for $t=1,...,p$.
It can be verified that the $\tau$th empirical quantile of $\{\widetilde{e}_{t,\tau}\}$ is zero.
Based on this fact and equation (\ref{seqc}), we obtain the estimation of $\rho_{k,\tau}$,
\[
r_{k,\tau}=\frac{1}{\sqrt{(\tau-\tau^2)\widetilde{\sigma}_e^2}}\cdot \frac{1}{n}\sum_{t=k+1}^n \psi_{\tau}(\widetilde{e}_{t,\tau})(\widetilde{e}_{t-k,\tau}-\widetilde{\mu}_e),
\]
where $k$ is a positive integer, $\widetilde{\mu}_e=n^{-1}\sum_{t=k+1}^n\widetilde{e}_{t,\tau}$, $\widetilde{\sigma}_e^2=n^{-1}\sum_{t=k+1}^n(\widetilde{e}_{t,\tau}-\widetilde{\mu}_e)^2$,  and the $\tau$th empirical quantile of $\{\widetilde{e}_{t,\tau}\}$ is zero. We name
$r_{k,\tau}$ the sample QACF of residuals.

Adapting the classical linear time series approach \citep{Li2004}, we examine the significance of $\{r_{k,\tau}\}$ individually and jointly.
For the given  positive integer $K$,
let $\mathbf{e}_{t-1,K}=(e_{t-1,\tau},...,e_{t-K,\tau})^{\prime}$, $\Sigma_{50}=E[\mathbf{e}_{t-1,K}\mathbf{z}_{t,p}^{*\prime}]$, $\Sigma_{51}=E[f_{t-1}(0)\mathbf{e}_{t-1,K}\mathbf{z}_{t,p}^{*\prime}]$, and
\[
\Omega_5=\frac{1}{\sigma_e^2}\{E(\mathbf{e}_{t-1,K}\mathbf{e}_{t-1,K}^{\prime}) +\Sigma_{51}\Sigma_{41}^{-1}\Sigma_{40}\Sigma_{41}^{-1}\Sigma_{51}^{\prime} -\Sigma_{51}\Sigma_{41}^{-1}\Sigma_{50}^{\prime}-\Sigma_{50}\Sigma_{41}^{-1}\Sigma_{51}^{\prime}\}.
\]
Then, we obtain the asymptotic distribution of $R_{\tau}=(r_{1,\tau},...,r_{K,\tau})^{\prime}$ given below.

\begin{thm}\label{thm5}
Assume that $0<\Sigma_{41}<\infty$, $\Sigma_{51}<\infty$, and Assumption \ref{assum3} holds. We then have
\[
\sqrt{n}R_{\tau}\rightarrow_dN(0,\Omega_5).
\]
\end{thm}

\noindent
Applying the same techniques as used in the estimate of $\Omega_3$,
we are able to estimate the asymptotic variance $\Omega_5$ and denote it
$\widehat{\Omega}_5$. In addition, let the $k$-th diagonal element of $\widehat{\Omega}_5$ be
$\widehat{\Omega}_{5k}$. Then, one can employ $r_{k,\tau}/\sqrt{\widehat{\Omega}_{5k}}$
to examine the significance of the $k$-th lag in the residual series.

To check the significance of $R_{\tau}$ jointly, it is natural to consider the test statistic $R_{\tau}^{\prime}\widehat{\Omega}_5^{-1}R_{\tau}$.
However, $\widehat{\Omega}_5$ may not be invertible. Hence, we  approximate
$\Omega_5$ by $I_K-\sigma_e^{-2}\Sigma_{50}\Sigma_{40}^{-1}\Sigma_{50}^{\prime}$, which
holds under the assumption that $\{e_{t,\tau}\}$ is an independent and identically distributed
($i.i.d.$) sequence and $f_{t-1}(0)$ is a constant. The resulting matrix
is idempotent and has  rank  $K-p$. This allows us to obtain
 a Box-Pierce type \citep{Box_Pierce1970} test statistic,
\[
Q_{BP}(K)=n\sum_{j=1}^Kr_{j,\tau}^2,
\]
which  follows an approximately
chi-squared distribution with $K-p$ degrees of freedom, $\chi_{K-p}^2$.
Accordingly, $Q_{BP}(K)$ can be used to test the significance of $\rho_{1,\tau}$ to
$\rho_{K,\tau}$ jointly.

\section{Simulations and an empirical example}

\subsection{Simulation studies}

We conduct five simulation experiments to assess the finite-sample performance of the proposed methods.
Specifically, the first simulation experiment is for the sample quantile correlation and the sample quantile partial correlation proposed in Section 2, and the last four experiments are, respectively,
 for identification, estimation, and diagnosis as  introduced in Section 3.
In all experiments, we conduct 1000 realizations for each combination of sample sizes $n=50$, 100, and 200 and quantiles, $\tau=0.25$, 0.50, and  0.75.

In the first simulation experiment, we generate the $i.i.d.$ samples $\{(X_i,Y_i,Z_i),i=1,...,n\}$ from the following multivariate normal distribution,
\[
(X,Y,Z)\sim N\left\{\mathbf{0},\left(\begin{array}{ccc} 1.0 &0.5&0.5\\ 0.5&1.0&0.5\\ 0.5&0.5&1.0 \end{array}\right)\right\}.
\]
After algebraic simplification, we obtain that  $$\textrm{qcor}_{\tau}\{Y,X\} =0.5\exp\{-0.5[\Phi^{-1}(\tau)]^2\}/\sqrt{(\tau-\tau^2)2\pi},$$ and $\textrm{qpcor}_{\tau}\{Y,X|Z\}=\textrm{qcor}_{\tau}\{Y,X\}/\sqrt{3}$, where $\Phi(\cdot)$ is the cumulative standard normal distribution.
Tables \ref{table1} and \ref{table2} present the bias (BIAS) and  estimated standard deviation (ESD), respectively, of the sample quantile correlations $\widehat{\textrm{qcor}}_{\tau}\{Y,X\}$ and the sample quantile partial correlations $\widehat{\textrm{qpcor}}_{\tau}\{Y,X|Z\}$,
calculated from  1000 realizations.

To estimate the asymptotic variances ${\Omega}_1$ and ${\Omega}_2$, we mainly need to
estimate the quantities ${\mu}_{X|Y}$ and ${\Sigma}_{20}$, addressed in Subsection 2.2.
To this end, we
employ the Nadaraya-Watson approach with the two bandwidth selection methods proposed by \cite{Bofinger1975}
and \cite{Hall_Sheather1988}, respectively, which are given below.
\[
h_B=n^{-1/5}\left\{\frac{4.5\phi^4(\Phi^{-1}(\tau))}{[2(\Phi^{-1}(\tau))^2+1]^2}\right\}^{1/5} \hspace{5mm}\text{and}\hspace{5mm}
h_{HS}=n^{-1/3}z_{\alpha}^{2/3}\left\{\frac{1.5\phi^2(\Phi^{-1}(\tau))}{2(\Phi^{-1}(\tau))^2+1}\right\}^{1/3},
\]
where $\phi(\cdot)$ is the standard normal density function, $z_{\alpha}=\Phi^{-1}(1-\alpha/2)$, for the construction of $1-\alpha$ confidence intervals, and $\alpha$ is set to 0.05. Furthermore, we consider
 two more bandwidths, $0.6h_B$ and $3h_{HS}$, suggested by \cite{Koenker_Xiao2006}. In sum, we have four bandwidth choices.
 The resulting asymptotic variance estimates, $\widehat{\Omega}_1$ and $\widehat{\Omega}_2$, yield
their correspondingly asymptotic standard deviations (ASDs) given in Tables \ref{table1}
 and \ref{table2}. Both tables indicate that the ASDs are close to their corresponding ESDs even when $n=50$,
and they become smaller as the same size gets larger.
In addition,  biases are close to zero, and decrease as the sample size increases.
Moreover, all four bandwidths lead to similar results, although  $3h_{HS}$ is slightly better than the others.

The second Monte Carlo experiment studies the performance of QPACF for identifying the order of the QAR models. we generate the data from the following process,
\begin{equation}\label{eqa}
y_t=0.1+0.5y_{t-1}+e_t,
\end{equation}
where $\{e_t\}$ is an $i.i.d$ sequence with standard normal distribution.
Under the above setting, it can be shown that $\Omega_3=1$.
We then  employ the approach of \cite{Hendricks_Koenker1991} with the four bandwidths used in the first experiment to estimate the density function, $f_{t-1}(0)$. This allows us to further estimate the variance matrix $\Omega_3$ in Theorem \ref{thm3} (see Subsection 3.1).
Table \ref{table3} presents  the bias and estimated standard deviation
of $\widetilde{\phi}_{kk,\tau}$  at $k=2$, 4, and 6.
It shows that biases are small even when $n=50$, and the ESDs are close to the ASDs as well as their theoretical value $1/\sqrt{n}$.

The third simulation experiment investigates the finite-sample performance of the QAR estimates.
We use the same data generated from  \eqref{eqa}, and then fit it with the QAR model
\eqref{section3_eq1} with $p=1$.
In addition, we employ the same approach as given in the second experiment to estimate $f_{t-1}(0)$.
As a result, the variance matrix $\Omega_4$ in Theorem \ref{thm4}
can be estimated (see Subsection 3.1).
Table \ref{table4} presents the biases, estimated standard deviations, and
asymptotic standard deviations of  parameter estimates $\widetilde{\phi}_{0}(\tau)$ and $\widetilde{\phi}_{1}(\tau)$.
It shows that biases are close to zero even when the sample size is as small as $n=50$. In addition, the ESDs are close to the ASDs, and both of them decrease as the sample size increases. Moreover, there is no discernible difference among the four bandwidths, although  $3h_{HS}$ often yields the smallest ASD.

The fourth simulation experiment examines the finite-sample performance of the sample QACF of residuals individually via the asymptotic result in Theorem \ref{thm5}.
All settings are the same as those in the third experiment.
Table \ref{table5} presents  the biases, estimated standard deviations, and asymptotic standard deviations of  $r_{k,\tau}$ at $k=1$, 3, and 5. Apparently, biases are small and the ASDs are close to
their corresponding ESDs.

Finally, the fifth experiment studies the approximate test statistic $Q_{BP}(K)$.
To this end, we generate data from the following process,
\[
y_t=0.5y_{t-1}+\phi y_{t-2}+e_t,
\]
where $\{e_t\}$ are $i.i.d.$ standard normal random variables. In addition,
$\phi=0$ corresponds to the null hypothesis, while $\phi>0$ is associated with the alternative hypothesis.
Moreover, the nominal level is  5\%.   Table \ref{table6} reports sizes and powers of  $Q_{BP}(K)$ with $K=6$. It shows that
$Q_{BP}(K)$ controls the size well, and its power increases quickly  when the sample size or $\phi$
becomes larger.
Consequently, the above six simulation studies  perform
satisfactorily and support our theoretical findings.

\subsection{Nasdaq Composite}

This example considers the log return (as a percentage) of the daily closing price on the Nasdaq Composite from January 1, 2002 to December 31, 2007.
There are 1,235 observations in total, and Figure \ref{nasdaq}
depicts the time series plot and the classical sample ACF. It is not surprising to conclude that
these returns (i.e., log returns) are uncorrelated and can be treated as an evidence in support of the fair market theory.
However, \cite{Veronesi1999} found that the stock markets under-react to good news in bad times and over-react to bad news in good times.  Hence, \cite{Baur_Dimpfl_Jung2012}
proposed aligning a good (bad) state with upper (lower) quantiles by fitting their stock returns data with the QAR(1) type models. This motivates us to  employ the general QAR model with our proposed methods to explore the dependence pattern of stock returns at a lower quantile ($\tau=0.2$), the median ($\tau=0.5$), and an upper quantile ($\tau=0.8$).

We first fit the returns at the lower quantile ($\tau=0.2$), and then present its sample QPACF in Panel A of Figure \ref{nasdaq_fit}. It shows that  lags 1, 2, and 13 are significant, which suggests QAR(13) could be considered for model fitting. We then refine the model via the backward variable selection procedure at the $5\%$ significance level. The resulting model is
\begin{equation}\label{naslow}
\widehat{Q}_{0.2}(y_t|\mathcal{F}_{t-1})= -0.4114_{0.0269}+0.1117_{0.0482}y_{t-1}+0.0951_{0.0471}y_{t-2}+0.0992_{0.0457}y_{t-13},
\end{equation}
where the subscripts of parameter estimates are their associated standard errors.
Accordingly, the above coefficients are all significant at the 5\% significance level.
In addition, the second graph in Panel A presents the sample QPACF of residuals,
and no lags stand out. This, together with the $p$-value of $Q_{BP}(18)$ being 0.742, implies that this model is adequate.

We next consider the scenario with $\tau=0.5$. The sample QPACF in Panel B
indicates that all lags are insignificant. Hence, we fit the following model,
\begin{equation}\label{nasmid}
\widehat{Q}_{0.5}(y_t|\mathcal{F}_{t-1})= 0.0036_{0.0170}.
\end{equation}
The above coefficient is not only small, but also not significant. In addition,
none of the lags in the sample QACF of residuals in Panel B show significance. Moreover,  the $p$-value of $Q_{BP}(18)$ is 0.566. Consequently, the above model is appropriate.

Finally, we study the upper quantile scenario with $\tau=0.8$.
The sample QPACF in Panel C exhibits  that  lags 1, 2, 7, 10 and 15 are significant, and suggests QAR(15) could be considered for model fitting.
After refining the model via the backward variable selection procedure,
we obtain
\begin{align}
\begin{split}\label{nashigh}
\widehat{Q}_{0.8}(y_t|\mathcal{F}_{t-1})=  0.3988_{0.0230}&-0.1076_{0.0371}y_{t-1}-0.0825_{0.0412}y_{t-2}\\
&-0.0790_{0.0361}y_{t-10}-0.0802_{0.0302}y_{t-15},
\end{split}
\end{align}
where all coefficients are significant at the $5\%$ significance level.
In addition,  the sample QACF of residuals in Panel C displays that all lags are insignificant.
This, in conjunction with  the $p$-value of $Q_{BP}(18)$ being  0.215, indicates that the above model
fits the data reasonably well.

Based on the  three fitted QAR models, (\ref{naslow}), (\ref{nasmid}), and (\ref{nashigh}), we obtain the following conclusions.
(i.) The lag coefficients at the lower quantile ($\tau=0.2$) are all positive.
This indicates that if the returns in past days have been positive (negative), then today's negative return  is alleviated (even lower). It also implies that stock markets
under-react to good news in bad times.
(ii.) The lag coefficients at the upper quantile ($\tau=0.8$) are all negative.
This shows that if the returns in  past days have been negative  (positive), then today's positive return
 ($\tau=0.8$) is even higher (dampened). As a result,  stock markets
over-react to bad news in good times.
(iii.) The intercept  at ($\tau=0.5$)  has a small value and is insignificant at the 5\% significance level. Thus, the conditional median of returns is almost zero as we expected. In addition,
equation (\ref{nasmid}) indicates that
today's return is not affected by the returns of recent past days.
Although we only report the results of the lower and higher quantiles at $\tau=0.2$ and $\tau=0.8$,
our studies yield the same conclusions across various lower and upper quantiles.
In sum, our proposed methods support Veronesi's (\citeyear{Veronesi1999}) equilibrium explanation for stock market reactions.

\section{Discussion}

In quantile regression models, we propose the quantile correlation and quantile partial correlation.
Then, we apply them to quantile autoregressive models, which yields
the quantile autocorrelation and quantile partial autocorrelation. In practice,
the response time series may depend on exogenous variables. Hence, it is of interest to extend those correlation measures
to the quantile autoregressive model with the exogenous variables given below.
\[
Q_{\tau}(y_t|\mathcal{F}_{t-1})=\phi_0(\tau)+\sum_{i=1}^p\phi_i(\tau)y_{t-i} +\sum_{j=1}^q\beta_j^{\prime}(\tau)\mathbf{x}_{t-j}, \hspace{2mm}\text{for}\hspace{2mm}0<\tau<1,
\]
where $\mathbf{x}_t$ is a vector of time series, and $\phi_i(\tau)$ and $\beta_j(\tau)$ are functions $[0,1]\rightarrow R$, see \cite{Galvao_Montes-Rojas_Park2012}.
In addition, the application of the proposed correlations to the quantile regression model with autoregressive errors is worth further investigation. Clearly, the contribution of the proposed measures is not limited to those two models. For example, variable screening and selection  (e.g., \citealt{Fan_Lv2008}; \citealt{Wang2009}) in quantile regressions are other important topics for future research.
In sum, this paper introduces valuable measures to broaden and facilitate the use of quantile models.

\renewcommand{\thesection}{A}
\section*{Appendix: technical proofs}

\begin{proof}[Proof of Lemma \ref{lem1}]
For $a,b\in R$, denote the function $h(a,b)=E[\rho_{\tau}(Y^*)]$, where $Y^*=Y-a-bX$. We first show that $h(a,b)$ is a continuously differentiable function and has derivatives,
\begin{equation*}
\frac{\partial h(a,b)}{\partial a}=-E[\psi_\tau(Y^*)]=P(Y^*<0)-\tau \hspace{5mm}\text{and}\hspace{5mm} \frac{\partial h(a,b)}{\partial b}=-E[\psi_\tau(Y^*)X].
\end{equation*}
 For $u\neq 0$,
\begin{align}
\begin{split}\label{proof_lem1_eq1}
\rho_{\tau}(u-v)-\rho_{\tau}(u)&=-v\psi_{\tau}(u)+\int_0^v[I(u\leq s)-I(u<0)]ds\\
 &=-v\psi_{\tau}(u)+(u-v)[I(0>u>v)-I(0<u<v)],
\end{split}
\end{align}
see \cite{Koenker_Xiao2006}.
This, together with H\"older's inequality and the fact that $|Y^*|/|X|$ is a continuous random variable,
leads to
\begin{align*}
|\frac{1}{c}&[h(a,b+c)-h(a,b)]+E[\psi_\tau(Y^*)X]|\\
 &=|\frac{1}{c}E[\rho_{\tau}(Y^*-cX)-\rho_{\tau}(Y^*)]+E[\psi_\tau(Y^*)X]|\\
 &=|\frac{1}{c}E\{(Y^*-cX)[I(0>Y^*>cX)-I(0<Y^*<cX)]\}|\\
 &\leq E[|X|I(|Y^*|<|c|\cdot|X|)]\leq (EX^2)^{1/2}[P(|Y^*|/|X|<|c|)]^{1/2},
\end{align*}
which tends to zero as $c\rightarrow 0$.
Accordingly,  $\partial h(a,b)/\partial b$ is obtained.
Analogously, we have  $\partial h(a,b)/\partial a$.
By H\"older's inequality, we can further prove the continuity of ${\partial h(a,b)}/{\partial b}$.
Moreover, the continuity of both $X$ and $Y$ implies that ${\partial h(a,b)}/{\partial a}$ is a continuous function.
It is noteworthy that
 $h(a,b)$ is a convex function with $\lim_{a^2+b^2\rightarrow\infty}h(a,b)=+\infty$.
This, in conjunction with the above results, demonstrates that
the values of $a_0$ and $b_0$ satisfy
\begin{equation}\label{proof_lem1_eq2}
E[\psi_\tau(Y-a_0-b_0X)]=0\hspace{5mm}\text{and}\hspace{5mm}E[\psi_\tau(Y-a_0-b_0X)X]=0.
\end{equation}

We next show the uniqueness of $(a_0,b_0)$. Suppose that there is another pair of values $(a_1,b_1)$ such that $h(a_1,b_1)=h(a_0,b_0)=\argmin_{a,b}E[\rho_{\tau}(Y-a-bX)]$. Let $Y_0=Y-a_0-b_0X$ and $\xi=(a_1-a_0)+(b_1-b_0)X$. Then, by \eqref{proof_lem1_eq1} and \eqref{proof_lem1_eq2},
\begin{align*}
0&=h(a_1,b_1)-h(a_0,b_0)=E[\rho_{\tau}(Y_0-\xi)-\rho_{\tau}(Y_0)]\\
 &=-E[\xi \psi_{\tau}(Y_0)]+E[(Y_0-\xi)I(0>Y_0>\xi)]+E[(\xi-Y_0)I(0<Y_0<\xi)]\\
 &=E[(Y_0-\xi)I(0>Y_0>\xi)]+E[(\xi-Y_0)I(0<Y_0<\xi)].
\end{align*}
Note that both $(Y_0-\xi)I(0>Y_0>\xi)$ and $(\xi-Y_0)I(0<Y_0<\xi)$ are nonnegative random variables, and $Y_0-\xi$ is a continuous random variable. Thus,  with probability one, $I(0>Y_0>\xi)=I(0<Y_0<\xi)=0$, which implies  $(a_1,b_1)= (a_0,b_0)$.

Finally,
 if $b_0=0$, then  \eqref{proof_lem1_eq2} leads to $a_0=Q_{\tau,Y}$ and $\textrm{qcov}_{\tau}\{Y,X\}=E[\psi_\tau(Y-a_0-b_0X)X]=0$. On the other hand, if $\textrm{qcov}_{\tau}\{Y,X\}=0$, then equation \eqref{proof_lem1_eq2} with $(a_0,b_0)=(Q_{\tau,Y},0)$ holds. By the uniqueness property, $b_0=0$, which
completes the proof.
\end{proof}

\begin{proof}[Proof of Lemma \ref{lem2}]
Let $ Y^*=Y-\alpha_2-\beta_2^{\prime}\mathbf{Z}$ and
\[
(\alpha_4,\beta_4^{\prime},\gamma_4)=\argmin_{\alpha,\beta,\gamma}E[\rho_{\tau}( Y^* -\alpha-\beta^{\prime}\mathbf{Z}-\gamma X)].
\]
Since the random vector $(X,Y^*,\mathbf{Z}^{\prime})^{\prime}$ has a joint density, we apply
 similar techniques to those in  the proof of Lemma \ref{lem1}  to show that
\begin{equation}\label{proof_lem2_eq1}
E[\psi_{\tau}(Y^*)]=0, \hspace{5mm} E[\psi_{\tau}(Y^*)\mathbf{Z}]=\mathbf{0},
\end{equation}
and the values of $\alpha_4$, $\beta_4^{\prime}$ and $\gamma_4$ are unique and satisfy
\begin{equation}\label{proof_lem2_eq2}
E[\psi_{\tau}(Y^*-\alpha_4-\beta_4^{\prime}\mathbf{Z}-\gamma_4X)(1,\mathbf{Z}^{\prime},X)^{\prime}]=\mathbf{0},
\end{equation}
where $\mathbf{0}$ is a $(q+2)\times 1$ zero vector.

From \eqref{proof_lem2_eq2}, if $(\alpha_4,\beta_4^{\prime},\gamma_4)^{\prime}=\mathbf{0}$, then $\textrm{qcov}_{\tau}\{Y^*,X\}=E[\psi_{\tau}(Y^*)X]=0$. On the other hand,  $\textrm{qcov}_{\tau}\{Y^*,X\}=0$, together with \eqref{proof_lem2_eq1}, implies that equation \eqref{proof_lem2_eq2} with $(\alpha_4,\beta_4^{\prime},\gamma_4)^{\prime}=\mathbf{0}$ holds. Accordingly, we have shown that $\textrm{qcov}_{\tau}\{Y^*,X\}=0$ if and only if $(\alpha_4,\beta_4^{\prime},\gamma_4)^{\prime}=\mathbf{0}$.
Based on the definitions of $(\alpha_2, \beta_2^{\prime})$ and $(\alpha_3,\beta_3^{\prime},\gamma_3)$
in Subsection 2.1,
we further have that
$\alpha_4=\alpha_3-\alpha_2$, $\beta_4=\beta_3-\beta_2$, and $\gamma_4=\gamma_3$.
Finally, using the fact that $\textrm{qpcor}_{\tau}\{Y,X|\mathbf{Z}\} =\textrm{qcov}_{\tau}\{Y^*,X\}/\sqrt{(\tau-\tau^2)\sigma_{X|\mathbf{Z}}^2}$ completes the proof.
\end{proof}

\begin{proof}[Proof of Lemma \ref{lem3}]
For $k=p$, let
\[
(\alpha_3,\beta_3^{\prime},\gamma_3)=\argmin_{\alpha,\beta,\gamma} E[\rho_{\tau}(y_{t}-\alpha-\beta^{\prime}\mathbf{z}_{t,p-1}-\gamma y_{t-p})].
\]
It is noteworthy that $(\alpha_3,\beta_3^{\prime},\gamma_3)=(\phi_0(\tau),\phi_1(\tau),...,\phi_p(\tau))$. Since $\phi_p(\tau)\neq 0$, we apply Lemma \ref{lem2} and are able to show that $\phi_{pp,\tau}\neq 0$.

Let
$e_{t,\tau}=y_t-\phi_0(\tau)-\phi_1(\tau)y_{t-1}-\cdots -\phi_p(\tau)y_{t-p}$.
By \eqref{section3_eq1}, $I(e_{t,\tau}>0)$ is independent of $y_{t-k}$
for any $k>0$. In addition,
$(\alpha_2,\beta_2^{\prime})=(\phi_0(\tau),\phi_1(\tau),...,\phi_p(\tau),\mathbf{0}^{\prime})$
for $k> p$, where $\mathbf{0}$ is $(k-p)\times 1$ vector. Hence, $\phi_{kk,\tau}=0$ for $k>p$.
\end{proof}


\begin{proof}[Proof of Theorem \ref{thm1}]
For $u\neq 0$, we have that
\[I(u-v<0)-I(u<0)=I(v>u>0)-I(v<u<0).\]
Using this result, we then obtain
\begin{equation}\label{thm1_eq1}
\frac{1}{n}\sum_{i=1}^n \psi_{\tau}(Y_i-\widehat{Q}_{\tau,Y})(X_i-\bar{X})
  =\frac{1}{n}\sum_{i=1}^n \psi_{\tau}(Y_i-{Q}_{\tau,Y})X_i+\frac{1}{n}A_n-\bar{X}\cdot \frac{1}{n}\sum_{i=1}^n\psi_{\tau}(Y_i-\widehat{Q}_{\tau,Y}),
\end{equation}
where $A_n=\sum_{i=1}^ng_{\tau}(Y_i,{Q}_{\tau,Y},\widehat{Q}_{\tau,Y})X_i $ and
\begin{align*}
g_{\tau}&(Y_i,{Q}_{\tau,Y},\widehat{Q}_{\tau,Y}) \\
&=\psi_{\tau}(Y_i-\widehat{Q}_{\tau,Y})-\psi_{\tau}(Y_i-{Q}_{\tau,Y}) =-[I(Y_i<\widehat{Q}_{\tau,Y})-I(Y_i<{Q}_{\tau,Y})]\\
 &=I(\widehat{Q}_{\tau,Y}-{Q}_{\tau,Y}<Y_i-{Q}_{\tau,Y}<0) -I(\widehat{Q}_{\tau,Y}-{Q}_{\tau,Y}>Y_i-{Q}_{\tau,Y}>0).
\end{align*}

It can be shown that
\[
|\frac{1}{n}\sum_{i=1}^n\psi_{\tau}(Y_i-\widehat{Q}_{\tau,Y})| =|\tau-\frac{1}{n}\sum_{i=1}^nI(Y_i-\widehat{Q}_{\tau,Y})| =|\tau-\frac{[n\tau]}{n}|\leq \frac{1}{n}.
\]
This, together with the  law of large numbers, implies the last term of \eqref{thm1_eq1} satisfying
\begin{equation}\label{thm1_eq2}
\bar{X}\cdot \frac{1}{n}\sum_{i=1}^n\psi_{\tau}(Y_i-\widehat{Q}_{\tau,Y})=O_p(n^{-1}).
\end{equation}

We next consider the second term on the right-hand side of \eqref{thm1_eq1}.
For any $v\in R$, denote
\[
\xi_n(v) =\frac{1}{\sqrt{n}}\sum_{i=1}^n\{g_{\tau}(Y_i,{Q}_{\tau,Y},{Q}_{\tau,Y}+n^{-1/2}v)-E[g_{\tau}(Y_i,{Q}_{\tau,Y},{Q}_{\tau,Y}+n^{-1/2}v)|X_i]\}X_i,
\]
where
\[
E[g_{\tau}(Y_i,{Q}_{\tau,Y},{Q}_{\tau,Y}+n^{-1/2}v)|X_i]= -\int_{{Q}_{\tau,Y}}^{{Q}_{\tau,Y}+n^{-1/2}v}f_{Y_i|X_i}(y)dy
\]
and $f_{Y_i|X_i}(\cdot)$ is the conditional density of $Y_i$ given $X_i$.
Then, by H\"older's inequality, we have that
\begin{align}
\begin{split}\label{thm1_eq3}
E[\xi_n(v)]^2&=E[g_{\tau}(Y_i,{Q}_{\tau,Y},{Q}_{\tau,Y}+n^{-1/2}v)X_i]^2 \\
 &\leq [P(|Y_i-{Q}_{\tau,Y}|<n^{-1/2}v)]^{1/2}[EX_i^4]^{1/2}=o(1).
\end{split}
\end{align}
After algebraic simplification, we further obtain
\begin{align*}
\sup_{|v_1-v|<\delta}&|\xi_n(v_1)-\xi_n(v)|\\ &\leq \sup_{|v_1-v|<\delta} \frac{1}{\sqrt{n}}\sum_{i=1}^n|\{g_{\tau}(v_1)-g_{\tau}(v)\}X_i| +E[|\{g_{\tau}(v_1)-g_{\tau}(v)\}X_i||X_i]\\
&=\frac{1}{\sqrt{n}}\sum_{i=1}^n|\{g_{\tau}(v_1^*)-g_{\tau}(v)\}X_i| +E[|\{g_{\tau}(v_1^*)-g_{\tau}(v)\}X_i||X_i],
\end{align*}
where $v_1^*$ takes the value of $v+\delta$ or $v-\delta$.
Hence,
\begin{align}
\begin{split}\label{thm1_eq4}
E&\sup_{|v_1-v|<\delta}|\xi_n(v_1)-\xi_n(v)|\\
& \leq 2\sqrt{n} E|\{g_{\tau}(v_1^*)-g_{\tau}(v)\}X_i| \\
&=2\sqrt{n}E\left|\int_{{Q}_{\tau,Y}+n^{-1/2}v}^{{Q}_{\tau,Y}+n^{-1/2}v_1^*}f_{Y_i|X_i}(y)dyX_i\right|\\
&\leq \delta\cdot 2E[\sup_{|y|\leq\pi}f_{Y_i|X_i}({Q}_{\tau,Y}+y)|X_i|],
\end{split}
\end{align}
where $|n^{-1/2}v|<\pi$ and $|n^{-1/2}v_1^*|<\pi$ when $n$ is large.
Both (\ref{thm1_eq3}) and (\ref{thm1_eq4}), in conjunction with the theorem's assumptions and the finite converging theorem,
imply that $E\sup_{|v|\leq M}|\xi_n(v)|=o(1)$ for any $M>0$.
In addition, applying the theorem in Section 2.5.1 of \cite{Serfling1980}, we  have
\[
\sqrt{n}(\widehat{Q}_{\tau,Y}-{Q}_{\tau,Y})=f_Y^{-1}({Q}_{\tau,Y})\cdot \frac{1}{\sqrt{n}}\sum_{i=1}^n\psi_{\tau}(Y_i-{Q}_{\tau,Y})+o_p(1)=O_p(1).
\]
Accordingly,
\begin{align}
\begin{split}\label{thm1_eq8}
\frac{1}{\sqrt{n}}A_n &= -\frac{1}{\sqrt{n}}\sum_{i=1}^n \int_{{Q}_{\tau,Y}}^{{Q}_{\tau,Y}+(\widehat{Q}_{\tau,Y}-{Q}_{\tau,Y})}f_{Y_i|X_i}(y)dyX_i +o_p(1)\\
 &=-(\widehat{Q}_{\tau,Y}-{Q}_{\tau,Y})\frac{1}{\sqrt{n}}\sum_{i=1}^n f_{Y_i|X_i}({Q}_{\tau,Y})X_i +o_p(1)\\
 &=-\frac{E[f_{Y_i|X_i}({Q}_{\tau,Y})X_i]}{f_Y({Q}_{\tau,Y})}\cdot \frac{1}{\sqrt{n}}\sum_{i=1}^n\psi_{\tau}(Y_i-{Q}_{\tau,Y})+o_p(1).
\end{split}
\end{align}

Subsequently, using  \eqref{thm1_eq1}, \eqref{thm1_eq2}, and \eqref{thm1_eq8}, we obtain that
\begin{align}
\begin{split}\label{thm1_eq6}
\sqrt{n}&\left[\frac{1}{n}\sum_{i=1}^n \psi_{\tau}(Y_i-\widehat{Q}_{\tau,Y})(X_i-\bar{X})- \textrm{qcov}_{\tau}\{Y,X\}\right]\\
&=\frac{1}{\sqrt{n}}\sum_{i=1}^n [ \psi_{\tau}(Y_i-\widehat{Q}_{\tau,Y})(X_i-\bar{X})-\textrm{qcov}_{\tau}\{Y,X\}]\\
  &=\frac{1}{\sqrt{n}}\sum_{i=1}^n [\psi_{\tau}(Y_i-{Q}_{\tau,Y})(X_i-\mu_{X|Y}) -\textrm{qcov}_{\tau}\{Y,X\}]+o_p(1),
\end{split}
\end{align}
where $\mu_{X|Y}$ is defined in Subsection 2.2.
Since
\[
\sqrt{n}(\bar{X}-\mu_X)^2=\frac{1}{\sqrt{n}}\left[\frac{1}{\sqrt{n}}\sum_{i=1}^n(X_i-\mu_X)\right]^2 =O_p(n^{-1/2}),
\]
 we further have that
\begin{equation}\label{thm1_eq7}
\sqrt{n}(\widehat{\sigma}_X^2-\sigma_X^2)=\frac{1}{\sqrt{n}}\sum_{i=1}^n[(X_i-\mu_X)^2-\sigma_X^2]+o_p(1).
\end{equation}
Moreover,  \eqref{thm1_eq6}, \eqref{thm1_eq7}, the central limit theorem, and the Cramer-Wold device,
lead to
\[
\sqrt{n}\left(\begin{array}{c}
\widehat{\sigma}_X^2-\sigma_X^2 \\
n^{-1}\sum_{i=1}^n \psi_{\tau}(Y_i-\widehat{Q}_{\tau,Y})(X_i-\bar{X})- \textrm{qcov}_{\tau}\{Y,X\}
\end{array}\right) \rightarrow_d N(0,\Sigma),
\]
where
\[
\Sigma=\left(\begin{array}{cc} \Sigma_{11} &\Sigma_{13} \\ \Sigma_{13} &\Sigma_{12}
\end{array}\right),
\]
and $\Sigma_{11}$, $\Sigma_{12}$, and $\Sigma_{13}$ are defined  in Subsection 2.2.
Finally, following the Delta method \citep[Chapter 3]{vanderVaart1998}, we complete the proof.

\end{proof}


\begin{proof}[Proof of Theorem \ref{thm2}]
We first consider the term $\widehat{\sigma}_{X|\mathbf{Z}}^2$ in $\widehat{\textrm{qpcor}}_{\tau}\{Y,X|\mathbf{Z}\}$.
Let $\mathbf{Z}^*_i=(1,\mathbf{Z}_i^{\prime})^{\prime}$, $X_i^*=X_i-\alpha_1-\beta_1^{\prime}\mathbf{Z}_i$, $\theta_1=(\alpha_1,\beta_1^{\prime})^{\prime}$ and $\widehat{\theta}_1=(\widehat{\alpha}_1,\widehat{\beta}_1^{\prime})^{\prime}$, where $(\alpha_1,\beta_1^{\prime})$  and
 $(\widehat{\alpha}_1,\widehat{\beta}_1^{\prime})$ are defined  in Subsections 2.1 and  2.2, respectively. By the assumptions of this theorem, we have that $\theta_1=[E(\mathbf{Z}^*_i\mathbf{Z}_i^{*\prime})]^{-1}E(\mathbf{Z}^*_iX_i)$, $E(\mathbf{Z}^*_iX_i^*)=\mathbf{0}$, and
\[
\sqrt{n}(\widehat{\theta}_1-\theta_1) =\left(\frac{1}{n}\sum_{i=1}^n\mathbf{Z}^*_i\mathbf{Z}_i^{*\prime}\right)^{-1} \cdot \frac{1}{\sqrt{n}}\sum_{i=1}^n\mathbf{Z}^*_iX_i^*=O_p(1).
\]
According  the law of large numbers, we then have that
\begin{align}
\begin{split}\label{thm2_eq1}
\widehat{\sigma}_{X|\mathbf{Z}}^2 &=\frac{1}{n}\sum_{i=1}^n(X_i-\widehat{\theta}_1^{\prime}\mathbf{Z}^*_i)^2 \\ &=\frac{1}{n}\sum_{i=1}^n(X_i-\theta_1^{\prime}\mathbf{Z}^*_i)^2 +(\widehat{\theta}_1-\theta_1)^{\prime} \left(\frac{1}{n}\sum_{i=1}^n\mathbf{Z}^*_i\mathbf{Z}_i^{*\prime}\right)(\widehat{\theta}_1-\theta_1) \\ &\hspace{15mm}-2(\widehat{\theta}_1-\theta_1)^{\prime}\left(\frac{1}{n}\sum_{i=1}^n\mathbf{Z}^*_iX_i^*\right)\\
&=\frac{1}{n}\sum_{i=1}^n(X_i-\theta_1^{\prime}\mathbf{Z}^*_i)^2+o_p(n^{-1/2}).
\end{split}
\end{align}

We next consider the numerator in $\widehat{\textrm{qpcor}}_{\tau}\{Y,X|\mathbf{Z}\}$.
For the sake of simplicity, let
$Y_i^*=Y_i-\alpha_2-\beta_2^{\prime}\mathbf{Z}_i=Y_i-\theta_2\mathbf{Z}^*_i$,
$\theta_2=(\alpha_2,\beta_2^{\prime})^{\prime}$, and $\widehat{\theta}_2=(\widehat{\alpha}_2,\widehat{\beta}_2^{\prime})^{\prime}$, where
$Y_i^*$ is defined in the proof of Lemma 2, and
$(\alpha_2,\beta_2^{\prime})$  and
 $(\widehat{\alpha}_2,\widehat{\beta}_2^{\prime})$ are defined  in Subsections 2.1 and  2.2, respectively.
Under the theorem's assumptions, we  employ similar techniques to those used in  the proof of Lemma \ref{lem1} and given in \cite{Koenker2005} to show that there exists a unique $\theta_2$ such that  $E[\psi_{\tau}(Y_i^*)\mathbf{Z}^*_i]=\mathbf{0}$ and
\begin{equation}\label{thm2_eq2}
\sqrt{n}(\widehat{\theta}_2-\theta_2) =\{E[f_{Y_i|\mathbf{Z}_i}(\theta_2^{\prime}\mathbf{Z}^*_i)\mathbf{Z}^*_i\mathbf{Z}_i^{*\prime}]\}^{-1} \cdot \frac{1}{\sqrt{n}}\sum_{i=1}^n\psi_{\tau}(Y_i^*)\mathbf{Z}^*_i+o_p(1).
\end{equation}

Using a similar method to that for obtaining \eqref{thm1_eq1}, we have that
\begin{equation}\label{thm2_eq3}
\frac{1}{n}\sum_{i=1}^n\psi_{\tau}(Y_i-\widehat{\theta}_2\mathbf{Z}^*_i)X_i =\frac{1}{n}\sum_{i=1}^n\psi_{\tau}(Y_i^*)X_i +\frac{1}{n}\sum_{i=1}^ng_{\tau}(Y_i,\mathbf{Z}_i,\theta_2,\widehat{\theta}_2)X_i,
\end{equation}
where
\begin{align*}
g_{\tau}(Y_i,\mathbf{Z}_i,\theta_2,\widehat{\theta}_2)&= \psi_{\tau}(Y_i-\widehat{\theta}_2\mathbf{Z}^*_i)-\psi_{\tau}(Y_i^*) =-[I(Y_i<\widehat{\theta}_2\mathbf{Z}^*_i)-I(Y_i<\theta_2\mathbf{Z}^*_i)]\\
 &= I[(\widehat{\theta}_2-\theta_2)^{\prime}\mathbf{Z}^*_i<Y^*_i<0] -I[(\widehat{\theta}_2-\theta_2)^{\prime}\mathbf{Z}^*_i>Y^*_i>0].
\end{align*}
For any $\mathbf{v}\in R^{q+1}$, let
\begin{align*}
\xi_n(\mathbf{v})=\frac{1}{\sqrt{n}}\sum_{i=1}^n& [g_{\tau}(Y_i,\mathbf{Z}_i,\theta_2,\theta_2+n^{-1/2}\mathbf{v})X_i
+\int_{\theta_2^{\prime}\mathbf{Z}^*_i}^{\theta_2^{\prime}\mathbf{Z}^*_i +n^{-1/2}\mathbf{v}^{\prime}\mathbf{Z}^*_i} f_{Y_i|\mathbf{Z}_i,X_i}(y)dyX_i].
\end{align*}
Applying similar techniques to those for obtaining \eqref{thm1_eq3} and \eqref{thm1_eq4}, we can demonstrate that
\begin{align*}
E[\xi_n(\mathbf{v})]^2 &=E[g_{\tau}(Y_i,\mathbf{Z}_i,\theta_2,\theta_2+n^{-1/2}\mathbf{v})X_i]^2 \\
&\leq \{P(|Y_i^*|\leq n^{-1/2}|\mathbf{v}^{\prime}\mathbf{Z}^*_i|)\}^{1/2}\cdot (EX_i^4)^{1/2} =o(1)
\end{align*}
and, for any $\delta>0$ and $\mathbf{v}_1\in R^{p+1}$,
\begin{align*}
E \sup_{\|\mathbf{v}_1-\mathbf{v}\|\leq\delta}|\xi_n(\mathbf{v}_1)-\xi_n(\mathbf{v})|\leq \delta \cdot 2 E[\sup_{|y|\leq\pi}f_{Y_i|\mathbf{Z}_i,X_i}(\theta_2^{\prime}\mathbf{Z}^*_i+y)|X_i|].
\end{align*}
This implies that $E \sup_{\|\mathbf{v}\|\leq M}|\xi_n(\mathbf{v})|=o(1)$ for any $M>0$.
Note that, by \eqref{thm2_eq2}, $\sqrt{n}(\widehat{\theta}_2-\theta_2)=O_p(1)$. As a result,
\begin{align*}
\frac{1}{\sqrt{n}}\sum_{i=1}^n&g_{\tau}(Y_i,\mathbf{Z}_i,\theta_2,\widehat{\theta}_2)X_i \\ &=-\frac{1}{\sqrt{n}}\sum_{i=1}^n\int_{\theta_2^{\prime}\mathbf{Z}^*_i}^{\theta_2^{\prime}\mathbf{Z}^*_i +(\widehat{\theta}_2-\theta_2)^{\prime}\mathbf{Z}^*_i} f_{Y_i|\mathbf{Z}_i,X_i}(y)dyX_i  +o_p(1)\\
&=-(\widehat{\theta}_2-\theta_2)^{\prime}\cdot\frac{1}{\sqrt{n}}\sum_{i=1}^n f_{Y_i|\mathbf{Z}_i,X_i}(\theta_2^{\prime}\mathbf{Z}^*_i)X_i\mathbf{Z}^*_i+o_p(1)\\
&= -\Sigma_{21}^{\prime} \Sigma_{22}^{-1} \cdot \frac{1}{\sqrt{n}}\sum_{i=1}^n\psi_{\tau}(Y_i^*)\mathbf{Z}^*_i+o_p(1),
\end{align*}
where $\Sigma_{21}=E\{f_{Y_i|\mathbf{Z}_i,X_i}(\theta_2^{\prime}\mathbf{Z}^*_i)X_i\mathbf{Z}^*_i\}$ and $\Sigma_{22}=E[f_{Y_i|\mathbf{Z}_i}(\theta_2^{\prime}\mathbf{Z}^*_i)\mathbf{Z}^*_i\mathbf{Z}_i^{*\prime}]$ are defined  in Subsection 2.2.
This, together with  \eqref{thm2_eq3}, results in
\begin{equation}\label{thm2_eq5}
\frac{1}{n}\sum_{i=1}^n\psi_{\tau}(Y_i-\widehat{\theta}_2\mathbf{Z}^*_i)X_i =\frac{1}{n}\sum_{i=1}^n\psi_{\tau}(Y_i-\theta_2\mathbf{Z}^*_i)(X_i -\Sigma_{21}^{\prime} \Sigma_{22}^{-1}\mathbf{Z}^*_i)+o_p(n^{-1/2}).
\end{equation}
Subsequently,
by \eqref{thm2_eq1}, \eqref{thm2_eq5}, the central limit theorem, and the Cramer-Wold device, we obtain that
\[
\sqrt{n}\left(\begin{array}{c}
\widehat{\sigma}_{X|\mathbf{Z}}^2-{\sigma}_{X|\mathbf{Z}}^2 \\
n^{-1}\sum_{i=1}^n\psi_{\tau}(Y_i-\widehat{\theta}_2\mathbf{Z}^*_i)X_i- E[\psi_{\tau}(Y-\theta_2^{\prime}\mathbf{Z}^*)X]
\end{array}\right) \rightarrow_d N(0,\Sigma_2),
\]
where
\[
\Sigma_2=\left(\begin{array}{cc} \Sigma_{23} &\Sigma_{25} \\ \Sigma_{25} &\Sigma_{24}
\end{array}\right),
\]
and $\Sigma_{23}$, $\Sigma_{24}$, and $\Sigma_{25}$ are defined as in Subsection 2.2.
Finally, following the Delta method \citep[Chapter 3]{vanderVaart1998}, we complete the proof.

\end{proof}

\begin{proof}[Proof of Theorem \ref{thm3}]
We first consider the term $\widetilde{\sigma}_{y|\mathbf{z}}^2$ in $\widetilde{\phi}_{kk,\tau}$.
Let $\mathbf{z}^*_{t,k-1}=(1,\mathbf{z}_{t,k-1}^{\prime})^{\prime}$. Since $Ey_t^2<\infty$ and $E[y_t-E(y_t|\mathcal{F}_{t-1})]^2>0$, the matrix $E(\mathbf{z}^*_{t,k-1}\mathbf{z}_{t,k-1}^{*\prime})$ is finite and positive definite. Analogous to \eqref{thm2_eq1}, we can show that
\begin{align}
\begin{split}\label{thm3_eq1}
\widetilde{\sigma}_{y|\mathbf{z}}^2 &=\frac{1}{n}\sum_{t=k+1}^n(y_{t-k}-\alpha_1-\beta_1^{\prime}\mathbf{z}_{t,k-1})^2+o_p(n^{-1/2})\\
& =E(y_{t-k}-\alpha_1-\beta_1^{\prime}\mathbf{z}_{t,k-1})^2+o_p(1).
\end{split}
\end{align}

We next study the numerator of  $\widetilde{\phi}_{kk,\tau}$.
Let $\theta_2=(\phi_0(\tau),\phi_1(\tau),...,\phi_p(\tau),\mathbf{0}^{\prime})^{\prime}$,
and $\widetilde{\theta}_2=(\widetilde{\alpha}_2,\widetilde{\beta}_2^{\prime})^{\prime}$, where
$\mathbf{0}$ is the $(k-p)\times 1$ vector defined in the proof of Lemma 3, and
 $\widetilde{\alpha}_2$ and $\widetilde{\beta}_2$ are defined in Subsection 3.1. It is noteworthy that the series $\{y_t\}$
is fitted by model \eqref{section3_eq1} with order $k-1$ and the true parameter vector $\theta_2$.
Accordingly, $e_{t,\tau}=y_t-\theta_2^{\prime}\mathbf{z}^*_{t,k-1}$ and the parameter estimate of $\theta_2$ is
 $\widetilde{\theta}_2$. Then, using \eqref{thm4_eq1} in the proof of Theorem \ref{thm4}, we obtain that
\[
\sqrt{n}(\widetilde{\theta}_2-\theta_2) =\{E[f_{t-1}(0)\mathbf{z}^*_{t,k-1}\mathbf{z}_{t,k-1}^{*\prime}]\}^{-1} \cdot\frac{1}{n}\sum_{t=k+1}^n\psi_{\tau}(e_{t,\tau}) \mathbf{z}^*_{t,k-1}+o_p(n^{-1/2}).
\]
Applying a similar approach to that used in obtaining \eqref{thm1_eq8}, and then using the above result,
we further have that
\begin{align}
\begin{split}\label{thm3_eq2}
\frac{1}{n}\sum_{t=k+1}^n& [\psi_{\tau}(y_t-\widetilde{\theta}_2^{\prime}\mathbf{z}^*_{t-k})- \psi_{\tau}(e_{t,\tau})] y_{t-k}\\
&=-\frac{1}{n}\sum_{t=k+1}^n \int_0^{(\widetilde{\theta}_2-\theta_2)^{\prime}\mathbf{z}^*_{t,k-1}} f_{t-1}(s)dsy_{t-k}+o_p(n^{-1/2})\\
&=-(\widetilde{\theta}_2-\theta_2)^{\prime} \cdot\frac{1}{n}\sum_{t=k+1}^n f_{t-1}(0)y_{t-k}\mathbf{z}^*_{t,k-1}+o_p(n^{-1/2})\\
&=-A_1^{\prime}\Sigma_{31}^{-1} \cdot\frac{1}{n}\sum_{t=k+1}^n\psi_{\tau}(e_{t,\tau}) \mathbf{z}^*_{t,k-1}+o_p(n^{-1/2}),
\end{split}
\end{align}
where $A_1$ and $\Sigma_{31}$ are defined as in Subsection 3.1.
Subsequently, using similar techniques to those for obtaining \eqref{thm1_eq1} and the result from
equation (\ref{thm3_eq2}),
we obtain that
\begin{align}
\begin{split}\label{thm3_eq4}
\frac{1}{n}\sum_{t=k+1}^n &\psi_{\tau}(y_t-\widetilde{\alpha}_2-\widetilde{\beta}_2^{\prime}\mathbf{z}_{t,k-1}) y_{t-k}\\
&=\frac{1}{n}\sum_{t=k+1}^n \psi_{\tau}(e_{t,\tau}) y_{t-k}+\frac{1}{n}\sum_{t=k+1}^n [\psi_{\tau}(y_t-\widetilde{\theta}_2^{\prime}\mathbf{z}^*_{t,k-1})- \psi_{\tau}(e_{t,\tau})] y_{t-k}\\
&=\frac{1}{n}\sum_{t=k+1}^n \psi_{\tau}(e_{t,\tau}) [y_{t-k}-A_1^{\prime}\Sigma_{31}^{-1}\mathbf{z}^*_{t,k-1}]+o_p(n^{-1/2}).
\end{split}
\end{align}
Equations \eqref{thm3_eq1} and (\ref{thm3_eq4}), together with  the central limit theorem for the martingale difference sequence, complete the proof of the asymptotic normality of $\widetilde{\phi}_{kk,\tau}$. From Lemma \ref{lem3}, we also have that ${\phi}_{kk,\tau}=0$.

\end{proof}

\begin{proof}[Proof of Theorem \ref{thm4}]
For any $\mathbf{v}\in R^{p+1}$, denote
\begin{align*}
Q(\mathbf{v})& =\sum_{t=p+1}^n\rho_{\tau}(y_t-(\bm{\phi}(\tau)+n^{-1/2}\mathbf{v})^{\prime}\mathbf{z}^*_{t,p}) -\sum_{t=p+1}^n\rho_{\tau}(y_t-\bm{\phi}^{\prime}(\tau)\mathbf{z}^*_{t,p})\\
&=\sum_{t=p+1}^n\rho_{\tau}(e_{t,\tau}^*-n^{-1/2}\mathbf{v}^{\prime}\mathbf{z}^*_{t,p}) -\sum_{t=p+1}^n\rho_{\tau}(e_{t,\tau}^*),
\end{align*}
where $e_{t,\tau}^*=y_t-\bm{\phi}^{\prime}(\tau)\mathbf{z}^*_{t,p}$.
Applying \eqref{proof_lem1_eq1} and  techniques similar to those in the proof of Theorem 3.1 in \cite{Koenker_Xiao2006}, we can show that
\begin{align*}
Q(\mathbf{v})& = -\mathbf{v}^{\prime}\cdot\frac{1}{\sqrt{n}}\sum_{t=p+1}^n\psi_{\tau}(e_{t,\tau}^*)\mathbf{z}^*_{t,p} +\sum_{t=p+1}^n\int_0^{n^{-1/2}\mathbf{v}^{\prime}\mathbf{z}^*_{t,p}}I(e_{t,\tau}^*\leq s)-I(e_{t,\tau}^*<0)ds\\
&= -\mathbf{v}^{\prime}\cdot\frac{1}{\sqrt{n}}\sum_{t=p+1}^n\psi_{\tau}(e_{t,\tau}^*)\mathbf{z}^*_{t,p} +\frac{1}{2}\mathbf{v}^{\prime}E[f_{t-1}(0)\mathbf{z}^*_{t,p}\mathbf{z}_{t,p}^{*\prime}]\mathbf{v}+o_p(1).
\end{align*}
Note that $Q(\mathbf{v})$ is a convex function with respect to $\mathbf{v}$. By \cite{Knight1998}, we then have the Bahadur representation as follows,
\begin{equation}\label{thm4_eq1}
\sqrt{n}\{\widetilde{\bm{\phi}}(\tau)-\bm{\phi}(\tau)\}= \{E[f_{t-1}(0)\mathbf{z}^*_{t,p}\mathbf{z}_{t,p}^{*\prime}]\}^{-1}\cdot \frac{1}{\sqrt{n}}\sum_{t=p+1}^n\psi_{\tau}(e_{t,\tau}^*)\mathbf{z}^*_{t,p}+o_p(1).
\end{equation}
This, in conjunction with the central limit theorem and the Cramer-Wold device, completes
the proof.
\end{proof}

\begin{proof}[Proof of Theorem \ref{thm5}]
Without loss of generality, we assume that ${\mathbf{z}}_{1,p}$ is observable. Then
\[
\widetilde{e}_{t,\tau}=y_t-\widetilde{\bm{\phi}}^{\prime}(\tau)\mathbf{z}^*_{t,p}=y_t-{\bm{\phi}}^{\prime}(\tau)\mathbf{z}^*_{t,p} -(\widetilde{\bm{\phi}}(\tau)-{\bm{\phi}}(\tau))^{\prime}\mathbf{z}^*_{t,p} =e_{t,\tau}-(\widetilde{\bm{\phi}}(\tau)-{\bm{\phi}}(\tau))^{\prime}\mathbf{z}^*_{t,p}
\]
for $1\leq t\leq n$.
We first consider the term $\widetilde{\sigma}_e^2$ in  $r_{k,\tau}$.
By the ergodic theorem and the fact that $\widetilde{\bm{\phi}}(\tau)-{\bm{\phi}}(\tau)=O_p(n^{-1/2})$, we can show that
\[
\widetilde{\mu}_e=\frac{1}{n}\sum_{t=k+1}^n\widetilde{e}_{t,\tau} =\frac{1}{n}\sum_{t=k+1}^ne_{t,\tau} -(\widetilde{\bm{\phi}}(\tau)-{\bm{\phi}}(\tau))^{\prime}\frac{1}{n}\sum_{t=k+1}^n\mathbf{z}^*_{t,p} =E(e_{t,\tau})+o_p(1),
\]
and
\begin{align}
\begin{split}\label{thm5_eq1}
\widetilde{\sigma}_e^2 &=\frac{1}{n}\sum_{t=k+1}^n(\widetilde{e}_{t,\tau}-\widetilde{\mu}_e)^2 =\frac{1}{n}\sum_{t=k+1}^n\widetilde{e}_{t,\tau}^2-\widetilde{\mu}_e^2 \\
&=\frac{1}{n}\sum_{t=k+1}^ne_{t,\tau}^2 -2(\widetilde{\bm{\phi}}(\tau)-{\bm{\phi}}(\tau))^{\prime} \cdot\frac{1}{n}\sum_{t=k+1}^ne_{t,\tau}\mathbf{z}^*_{t,p}\\
&\hspace{10mm}+(\widetilde{\bm{\phi}}(\tau)-{\bm{\phi}}(\tau))^{\prime} \cdot\frac{1}{n}\sum_{t=k+1}^n\mathbf{z}^*_{t,p} \mathbf{z}_{t,p}^{*\prime} \cdot(\widetilde{\bm{\phi}}(\tau)-{\bm{\phi}}(\tau)) -\widetilde{\mu}_e^2\\
&=\sigma_e^2+o_p(1),
\end{split}
\end{align}
where $\sigma_e^2$ is defined in Subsection 3.2.

We next consider the numerator of $r_{k,\tau}$.
Using the fact that $|\sum_{t=k+1}^n \psi_{\tau}(\widetilde{e}_{t,\tau})|<1$, we obtain
\begin{align}
\begin{split}\label{thm5_eq2}
\frac{1}{n}\sum_{t=k+1}^n &\psi_{\tau}(\widetilde{e}_{t,\tau})(\widetilde{e}_{t-k,\tau}-\widetilde{\mu}_e) \\
&=\frac{1}{n}\sum_{t=k+1}^n \psi_{\tau}(y_t-\widetilde{\bm{\phi}}^{\prime}(\tau)\mathbf{z}^*_{t,p}) [e_{t-k,\tau}-(\widetilde{\bm{\phi}}(\tau)-{\bm{\phi}}(\tau))^{\prime}\mathbf{z}^*_{t-k,p}] +O_p(n^{-1})\\
&=\frac{1}{n}\sum_{t=k+1}^n \psi_{\tau}(y_t-\widetilde{\bm{\phi}}^{\prime}(\tau)\mathbf{z}^*_{t,p})e_{t-k,\tau} \\ &\hspace{5mm}-(\widetilde{\bm{\phi}}(\tau)-{\bm{\phi}}(\tau))^{\prime} \cdot\frac{1}{n}\sum_{t=k+1}^n \psi_{\tau}(y_t-\widetilde{\bm{\phi}}^{\prime}(\tau)\mathbf{z}^*_{t,p})\mathbf{z}^*_{t-k,p} +o_p(n^{-1/2}).
\end{split}
\end{align}
Applying similar techniques to those used in obtaining \eqref{thm1_eq8}, we are able to show that
\begin{align*}
\frac{1}{n}&\sum_{t=k+1}^n [\psi_{\tau}(y_t-\widetilde{\bm{\phi}}^{\prime}(\tau)\mathbf{z}^*_{t,p})- \psi_{\tau}(e_{t,\tau})]e_{t-k,\tau} \\
&=-\Sigma_{51,k}\Sigma_{41}^{-1}\cdot\frac{1}{n}\sum_{t=k+1}^n\psi_{\tau}(e_{t,\tau})\mathbf{z}^*_{t,p}+o_p(n^{-1/2}),
\end{align*}
where $\Sigma_{41}$ is defined in Subsection 3.1 and $\Sigma_{51,k}=E[f_{t-1}(0)e_{t-k,\tau}\mathbf{z}_{t,p}^{*\prime}]$.
In addition, using similar techniques to those in obtaining  \eqref{thm1_eq1} and the above result, we further  obtain that
\begin{align*}
\frac{1}{n}\sum_{t=k+1}^n \psi_{\tau}(y_t-\widetilde{\bm{\phi}}^{\prime}(\tau)\mathbf{z}^*_{t,p})e_{t-k,\tau} =\frac{1}{n}\sum_{t=k+1}^n\psi_{\tau}(e_{t,\tau})[e_{t-k,\tau} -\Sigma_{51,k}\Sigma_{41}^{-1}\mathbf{z}^*_{t,p}]+o_p(n^{-1/2}).
\end{align*}
 Analogously, we can verify that
\[
\frac{1}{n}\sum_{t=k+1}^n \psi_{\tau}(y_t-\widetilde{\bm{\phi}}^{\prime}(\tau)\mathbf{z}^*_{t,p})\mathbf{z}^*_{t-k,p} =O_p(n^{-1/2}).
\]
The above results, together with \eqref{thm5_eq1}, (\ref{thm5_eq2}), and the fact that $\widetilde{\bm{\phi}}(\tau)-{\bm{\phi}}(\tau)=O_p(n^{-1/2})$, imply
\[
r_{k,\tau}=\frac{1}{\sqrt{(\tau-\tau^2)\sigma_e^2}}\cdot \frac{1}{n}\sum_{t=k+1}^n\psi_{\tau}(e_{t,\tau})[e_{t-k,\tau} -\Sigma_{51,k}\Sigma_{41}^{-1}\mathbf{z}^*_{t,p}]+o_p(n^{-1/2}),
\]
and
\[
R_{\tau}=\frac{1}{\sqrt{(\tau-\tau^2)\sigma_e^2}}\cdot \frac{1}{n}\sum_{t=k+1}^n\psi_{\tau}(e_{t,\tau})[\mathbf{e}_{t-1,K} -\Sigma_{51}\Sigma_{41}^{-1}\mathbf{z}^*_{t,p}]+o_p(n^{-1/2}),
\]
where $\mathbf{e}_{t-1,K}$ and $\Sigma_{51}$ are defined in Subsection 3.2.
Subsequently, applying the central limit theorem for the martingale difference sequence and the Cramer-Wold device, we complete the proof.

\end{proof}

\newpage
\renewcommand{\baselinestretch}{1.3}

\begin{table}
\begin{center}
\caption{\label{table1} Bias (BIAS), estimated standard deviation (ESD), and asymptotic standard deviation
(ASD) of the sample quantile correlation $\widehat{\textrm{qcor}}_{\tau}\{Y,X\}$.}
\begin{tabular}{ccrccccc}
\\
\hline
$n$ & $\tau$ & BIAS & ESD & \multicolumn{4}{c}{ASD}\\ \cline{5-8}
    &        &      &              &$h_{HS}$ & $h_{B}$ & $3h_{HS}$ & $0.6h_{B}$\\
\hline
 50& 0.25 & 0.0062 & 0.1216 & 0.1274 & 0.1268 & 0.1199 & 0.1337 \\
   & 0.50 &-0.0055 & 0.1180 & 0.1216 & 0.1210 & 0.1167 & 0.1253\\
   & 0.75 & 0.0023 & 0.1221 & 0.1286 & 0.1280 & 0.1199 & 0.1364\\
100& 0.25 &-0.0045 & 0.0828 & 0.0867 & 0.0860 & 0.0836 & 0.0891\\
   & 0.50 &-0.0032 & 0.0792 & 0.0835 & 0.0829 & 0.0816 & 0.0847 \\
   & 0.75 &-0.0029 & 0.0824 & 0.0873 & 0.0863 & 0.0836 & 0.0897\\
200& 0.25 &-0.0012 & 0.0598 & 0.0601 & 0.0596 & 0.0586 & 0.0607\\
   & 0.50 &-0.0000 & 0.0585 & 0.0580 & 0.0577 & 0.0571 & 0.0582\\
   & 0.75 & 0.0005 & 0.0562 & 0.0601 & 0.0596 & 0.0586 & 0.0608 \\
      \hline
\end{tabular}
\end{center}
\end{table}

\begin{table}
\begin{center}
\caption{\label{table2} Bias (BIAS), estimated standard deviation (ESD), and asymptotic standard deviation (ASD)
of the sample quantile partial correlation $\widehat{\textrm{qpcor}}_{\tau}\{Y,X|Z\}$.}
\begin{tabular}{ccrccccc}
\\
\hline
$n$ & $\tau$ & BIAS & ESD & \multicolumn{4}{c}{ASD}\\ \cline{5-8}
    &        &      &              &$h_{HS}$ & $h_{B}$ & $3h_{HS}$ & $0.6h_{B}$\\
\hline
 50& 0.25 &-0.0135 & 0.1330 & 0.1383 & 0.1379 & 0.1281 & 0.1492\\
   & 0.50 &-0.0054 & 0.1365 & 0.1350 & 0.1337 & 0.1282 & 0.1398\\
   & 0.75 & 0.0029 & 0.1378 & 0.1407 & 0.1401 & 0.1299 & 0.1487\\
100& 0.25 &-0.0094 & 0.0901 & 0.0972 & 0.0963 & 0.0922 & 0.1004\\
   & 0.50 &-0.0026 & 0.0931 & 0.0943 & 0.0935 & 0.0912 & 0.0959\\
   & 0.75 & 0.0046 & 0.0971 & 0.0974 & 0.0961 & 0.0921 & 0.1008  \\
200& 0.25 &-0.0052 & 0.0663 & 0.0677 & 0.0669 & 0.0651 & 0.0688\\
   & 0.50 &-0.0004 & 0.0654 & 0.0660 & 0.0654 & 0.0646 & 0.0664\\
   & 0.75 & 0.0022 & 0.0655 & 0.0677 & 0.0669 & 0.0654 & 0.0689\\
\hline
\end{tabular}
\end{center}
\end{table}

\begin{table}
\begin{center}
\caption{\label{table3} Bias (BIAS), estimated standard deviation (ESD), and asymptotic standard deviation (ASD)
of the sample QPACF  of the observed time series, $\widetilde{\phi}_{kk,\tau}$,
at lags $k=2$, 4, and 6.}
\begin{tabular}{ccrccccc}
\\
\hline
$n$ & $\tau$ & BIAS & ESD & \multicolumn{4}{c}{ASD}\\ \cline{5-8}
    &        &      &              &$h_{HS}$ & $h_{B}$ & $3h_{HS}$ & $0.6h_{B}$\\
\hline
&&\multicolumn{6}{c}{$k=2$}\\
 50& 0.25 &-0.0280& 0.1417& 0.1419& 0.1419& 0.1416& 0.1422\\
   & 0.50 &-0.0343& 0.1439& 0.1416& 0.1416& 0.1416& 0.1417\\
   & 0.75 &-0.0316& 0.1485& 0.1418& 0.1419& 0.1416& 0.1421\\
100& 0.25 &-0.0163& 0.1021& 0.1001& 0.1001& 0.1001& 0.1002\\
   & 0.50 &-0.0168& 0.1042& 0.1001& 0.1001& 0.1001& 0.1001\\
   & 0.75 &-0.0102& 0.1009& 0.1001& 0.1001& 0.1001& 0.1002\\
200& 0.25 &-0.0107& 0.0732& 0.0707& 0.0707& 0.0707& 0.0707\\
   & 0.50 &-0.0092& 0.0711& 0.0707& 0.0707& 0.0707& 0.0707\\
   & 0.75 &-0.0077& 0.0728& 0.0707& 0.0707& 0.0707& 0.0707\\
&&\multicolumn{6}{c}{$k=4$}\\
 50& 0.25 &-0.0344& 0.1434& 0.1438& 0.1439& 0.1421& 0.1457\\
   & 0.50 &-0.0340& 0.1471& 0.1427& 0.1424& 0.1427& 0.1437\\
   & 0.75 &-0.0317& 0.1497& 0.1438& 0.1439& 0.1421& 0.1456\\
100& 0.25 &-0.0143& 0.1032& 0.1007& 0.1005& 0.1002& 0.1013\\
   & 0.50 &-0.0172& 0.1013& 0.1003& 0.1002& 0.1002& 0.1005\\
   & 0.75 &-0.0196& 0.1038& 0.1007& 0.1005& 0.1002& 0.1011\\
200& 0.25 &-0.0042& 0.0709& 0.0709& 0.0708& 0.0708& 0.0710\\
   & 0.50 &-0.0066& 0.0720& 0.0708& 0.0708& 0.0708& 0.0708\\
   & 0.75 &-0.0072& 0.0703& 0.0709& 0.0708& 0.0708& 0.0710\\
&&\multicolumn{6}{c}{$k=6$}\\
 50& 0.25 &-0.0278& 0.1486& 0.1489& 0.1489& 0.1483&0.1496\\
   & 0.50 &-0.0356& 0.1500& 0.1452& 0.1450& 0.1452& 0.1463\\
   & 0.75 &-0.0296& 0.1588& 0.1531& 0.1533& 0.1523& 0.1557\\
100& 0.25 &-0.0124& 0.1052& 0.1018& 0.1015& 0.1004& 0.1030\\
   & 0.50 &-0.0197& 0.1049& 0.1006& 0.1004& 0.1005& 0.1014\\
   & 0.75 &-0.0189& 0.1073& 0.1017& 0.1012& 0.1004& 0.1030\\
200& 0.25 &-0.0103& 0.0741& 0.0712& 0.0710& 0.0708& 0.0716\\
   & 0.50 &-0.0112& 0.0736& 0.0709& 0.0708& 0.0708& 0.0710\\
   & 0.75 &-0.0105& 0.0727& 0.0712& 0.0710& 0.0708& 0.0715\\
\hline
\end{tabular}
\end{center}
\end{table}

\begin{table}
\begin{center}
\caption{\label{table4} Bias (BIAS), estimated standard deviation (ESD), and asymptotic standard deviation
(ASD) of parameter estimates $\widetilde{\phi}_0(\tau)$ and $\widetilde{\phi}_1(\tau)$.}
\begin{tabular}{cccrccccc}
\\
\hline
$n$ & $\tau$ & Coefficients & BIAS & ESD & \multicolumn{4}{c}{ASD}\\ \cline{6-9}
    &        &&      &              &$h_{HS}$ & $h_{B}$ & $3h_{HS}$ & $0.6h_{B}$\\
\hline
 50& 0.25 & $\widetilde{\phi}_0(\tau)$ & 0.0065 & 0.2114 & 0.2285 & 0.2338 & 0.1846 & 0.2086\\
   &      & $\widetilde{\phi}_1(\tau)$ &-0.0487 & 0.1733 & 0.1907 & 0.1943 & 0.1601 & 0.1670\\
   & 0.50 & $\widetilde{\phi}_0(\tau)$ &-0.0129 & 0.1923 & 0.1957 & 0.2016 & 0.1942 & 0.1879\\
   &      & $\widetilde{\phi}_1(\tau)$ &-0.0471 & 0.1653 & 0.1684 & 0.1736 & 0.1664 & 0.1583\\
   & 0.75 & $\widetilde{\phi}_0(\tau)$ &-0.0229 & 0.2141 & 0.2287 & 0.2335 & 0.1845 & 0.2086\\
   &      & $\widetilde{\phi}_1(\tau)$ &-0.0495 & 0.1790 & 0.1899 & 0.1938 & 0.1590 & 0.1693\\
100& 0.25 & $\widetilde{\phi}_0(\tau)$ & 0.0016 & 0.1396 & 0.1513 & 0.1553 & 0.1282 & 0.1457\\
   &      & $\widetilde{\phi}_1(\tau)$ &-0.0187 & 0.1208 & 0.1264 & 0.1302 & 0.1102 & 0.1189\\
   & 0.50 & $\widetilde{\phi}_0(\tau)$ &-0.0039 & 0.1290 & 0.1329 & 0.1370 & 0.1363 & 0.1302\\
   &      & $\widetilde{\phi}_1(\tau)$ &-0.0186 & 0.1151 & 0.1133 & 0.1176 & 0.1169 & 0.1092\\
   & 0.75 & $\widetilde{\phi}_0(\tau)$ &-0.0118 & 0.1445 & 0.1488 & 0.1534 & 0.1269 & 0.1434\\
   &      & $\widetilde{\phi}_1(\tau)$ &-0.0197 & 0.1227 & 0.1241 & 0.1289 & 0.1091 & 0.1164\\
200& 0.25 & $\widetilde{\phi}_0(\tau)$ &-0.0010 & 0.0970 & 0.1019 & 0.1052 & 0.0911 & 0.1001\\
   &      & $\widetilde{\phi}_1(\tau)$ &-0.0068 & 0.0839 & 0.0860 & 0.0895 & 0.0785 & 0.0831\\
   & 0.50 & $\widetilde{\phi}_0(\tau)$ &-0.0046 & 0.0916 & 0.0916 & 0.0943 & 0.1049 & 0.0907\\
   &      & $\widetilde{\phi}_1(\tau)$ &-0.0110 & 0.0765 & 0.0782 & 0.0813 & 0.1082 & 0.0769\\
   & 0.75 & $\widetilde{\phi}_0(\tau)$ &-0.0099 & 0.0999 & 0.1021 & 0.1055 & 0.0912 & 0.1002\\
   &      & $\widetilde{\phi}_1(\tau)$ &-0.0104 & 0.0840 & 0.0860 & 0.0897 & 0.0787 & 0.0833\\
      \hline
\end{tabular}
\end{center}
\end{table}

\begin{table}
\begin{center}
\caption{\label{table5} Bias (BIAS), estimated standard deviation (ESD) and asymptotic standard deviation (ASD)
of the sample QACF of residuals, $r_{k,\tau}$, at $k=1$, 3, and 5.}
\begin{tabular}{ccrccccc}
\\
\hline
$n$ & QACF  & BIAS & ESD & \multicolumn{4}{c}{ASD}\\ \cline{5-8}
    &     &      &              &$h_{HS}$ & $h_{B}$ & $3h_{HS}$ & $0.6h_{B}$\\
\hline
   &   &\multicolumn{6}{c}{$\tau=0.25$}\\
 50& $r_{1,\tau}$ & 0.0195 & 0.0849 & 0.0704 & 0.0705 & 0.0699 & 0.0706\\
   & $r_{3,\tau}$ &-0.0011 & 0.1485 & 0.1478 & 0.1479 & 0.1476 & 0.1482\\
   & $r_{5,\tau}$ & 0.0002 & 0.1478 & 0.1497 & 0.1496 & 0.1495 & 0.1499\\
100& $r_{1,\tau}$ & 0.0075 & 0.0577 & 0.0490 & 0.0491 & 0.0489 & 0.0492\\
   & $r_{3,\tau}$ &-0.0042 & 0.1009 & 0.1008 & 0.1008 & 0.1007 & 0.1009\\
   & $r_{5,\tau}$ &-0.0025 & 0.1041 & 0.1027 & 0.1026 & 0.1027 & 0.1029\\
200& $r_{1,\tau}$ & 0.0035 & 0.0370 & 0.0351 & 0.0352 & 0.0351 & 0.0352\\
   & $r_{3,\tau}$ & 0.0008 & 0.0705 & 0.0701 & 0.0702 & 0.0701 & 0.0702\\
   & $r_{5,\tau}$ & 0.0005 & 0.0716 & 0.0716 & 0.0715 & 0.0716 & 0.0716\\
   &   &\multicolumn{6}{c}{$\tau=0.50$}\\
 50& $r_{1,\tau}$ & 0.0156 & 0.0779 & 0.0699 & 0.0700 & 0.0699 & 0.0704\\
   & $r_{3,\tau}$ &-0.0199 & 0.1409 & 0.1473 & 0.1473 & 0.1473 & 0.1475\\
   & $r_{5,\tau}$ &-0.0234 & 0.1463 & 0.1492 & 0.1491 & 0.1492 & 0.1493\\
100& $r_{1,\tau}$ & 0.0094 & 0.0553 & 0.0493 & 0.0493 & 0.0493 & 0.0494\\
   & $r_{3,\tau}$ &-0.0089 & 0.0998 & 0.1008 & 0.1007 & 0.1007 & 0.1008\\
   & $r_{5,\tau}$ &-0.0108 & 0.1022 & 0.1025 & 0.1025 & 0.1025 & 0.1026\\
200& $r_{1,\tau}$ & 0.0021 & 0.0366 & 0.0351 & 0.0351 & 0.0351 & 0.0351\\
   & $r_{3,\tau}$ &-0.0036 & 0.0693 & 0.0701 & 0.0701 & 0.0701 & 0.0701\\
   & $r_{5,\tau}$ &-0.0049 & 0.0703 & 0.0715 & 0.0715 & 0.0715 & 0.0716\\
   &   &\multicolumn{6}{c}{$\tau=0.75$}\\
 50& $r_{1,\tau}$ & 0.0112 & 0.0918 & 0.0704 & 0.0702 & 0.0698 & 0.0708\\
   & $r_{3,\tau}$ &-0.0355 & 0.1461 & 0.1477 & 0.1477 & 0.1474 & 0.1482\\
   & $r_{5,\tau}$ &-0.0390 & 0.1558 & 0.1495 & 0.1497 & 0.1493 & 0.1500\\
100& $r_{1,\tau}$ & 0.0091 & 0.0577 & 0.0495 & 0.0495 & 0.0494 & 0.0496\\
   & $r_{3,\tau}$ &-0.0181 & 0.1027 & 0.1008 & 0.1008 & 0.1007 & 0.1009\\
   & $r_{5,\tau}$ &-0.0199 & 0.1022 & 0.1026 & 0.1026 & 0.1025 & 0.1027\\
200& $r_{1,\tau}$ & 0.0018 & 0.0363 & 0.0351 & 0.0351 & 0.0351 & 0.0351\\
   & $r_{3,\tau}$ &-0.0082 & 0.0705 & 0.0701 & 0.0701 & 0.0701 & 0.0702\\
   & $r_{5,\tau}$ &-0.0094 & 0.0724 & 0.0716 & 0.0716 & 0.0716 & 0.0716\\
      \hline
\end{tabular}
\end{center}
\end{table}

\begin{table}
\begin{center}
\caption{\label{table6} Rejection rate of the test statistic $Q_{BP}(K)$ with $K=6$
and the 5\% nominal significance level.}
\begin{tabular}{cccccccc}
\\
\hline
&&\multicolumn{3}{c}{$\tau$}\\
\cline{3-5}
$n$ & $\phi$ & 0.25&0.50&0.75\\
\hline
 50   & 0.0 &0.052&0.046&0.056\\
      & 0.2 &0.078&0.067&0.082\\
      & 0.4 &0.221&0.249&0.231\\
 100  & 0.0 &0.041&0.056&0.052\\
      & 0.2 &0.129&0.146&0.126\\
      & 0.4 &0.532&0.602&0.514\\
 200  & 0.0 &0.048&0.051&0.051\\
      & 0.2 &0.283&0.325&0.257\\
      & 0.4 &0.891&0.952&0.886\\
\hline
\end{tabular}
\end{center}
\end{table}

\begin{figure}
\centering \scalebox{0.7}[0.7]{
\includegraphics{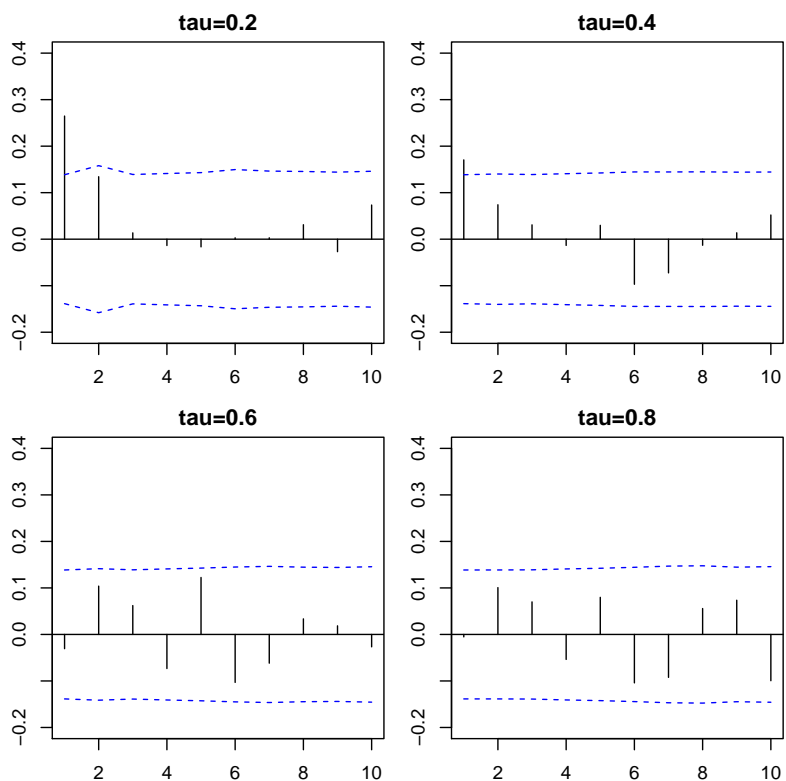}}
\caption{\label{fig1} The sample QPACF of the observed time series, $\widetilde{\phi}_{kk,\tau}$, with $\tau=0.2$, 0.4, 0.6, and 0.8. The dashed lines correspond to
$\pm 1.96\sqrt{\widehat{\Omega}_3/n}$.}
\end{figure}

\begin{figure}
\centering \scalebox{0.7}[0.7]{
\includegraphics{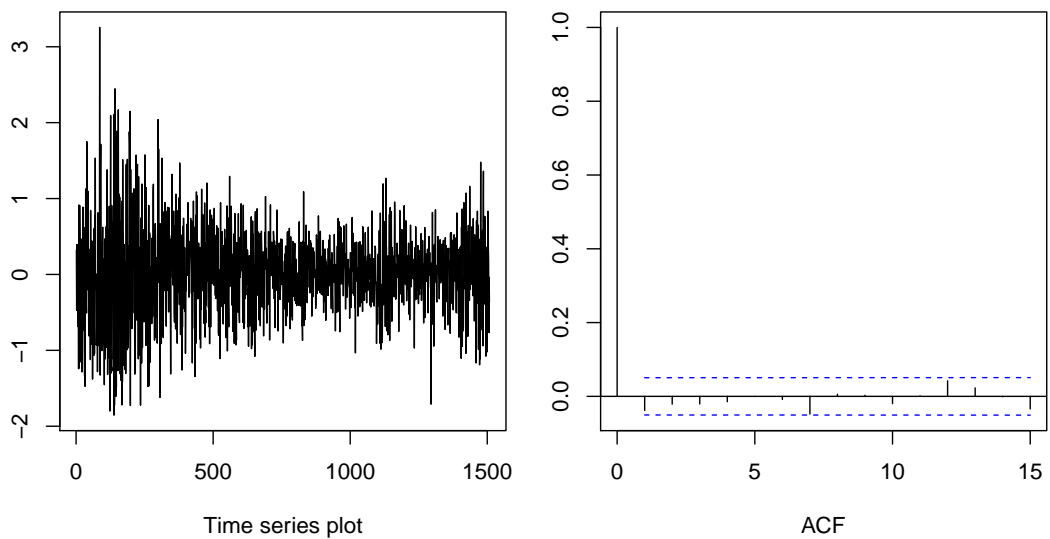}}
\caption{\label{nasdaq} The time series plot and the sample ACF of the log return (as a percentage) of the daily closing price on the Nasdaq Composite from January 1, 2002 to December 31, 2007.}
\end{figure}

\begin{figure}
\centering 
\includegraphics[scale=0.6]{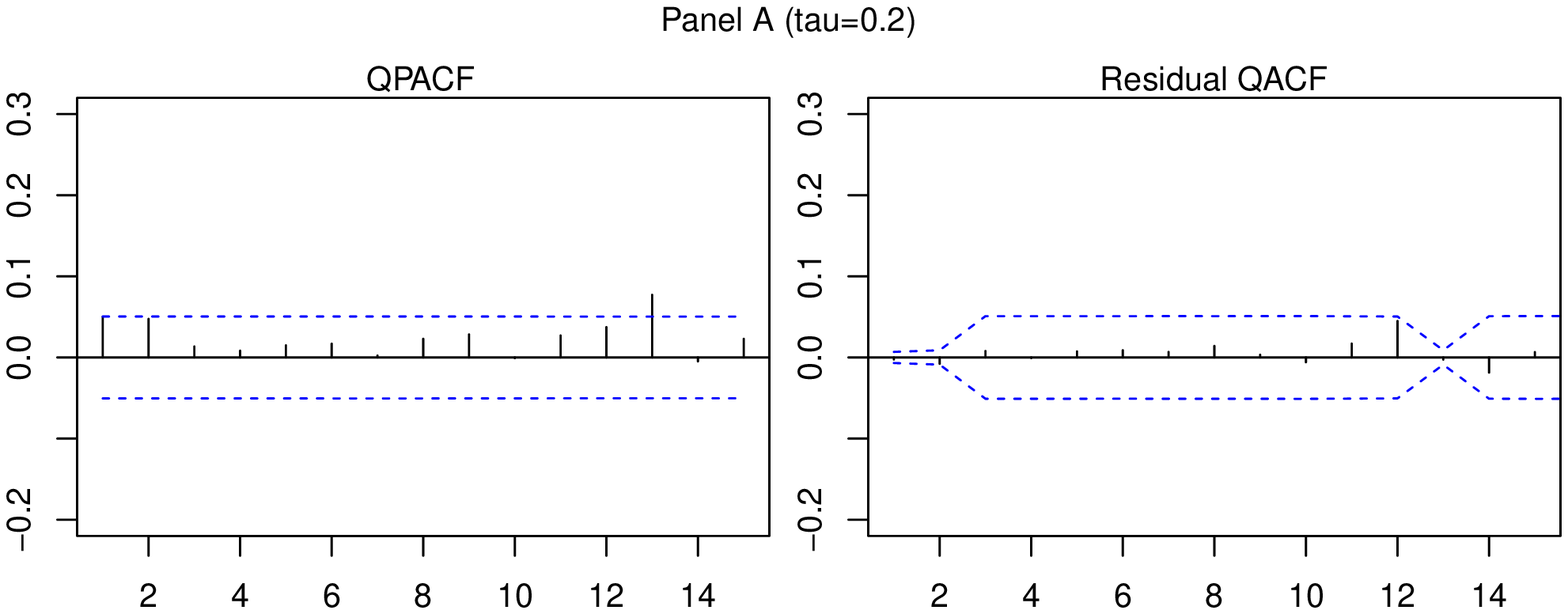}
\includegraphics[scale=0.6]{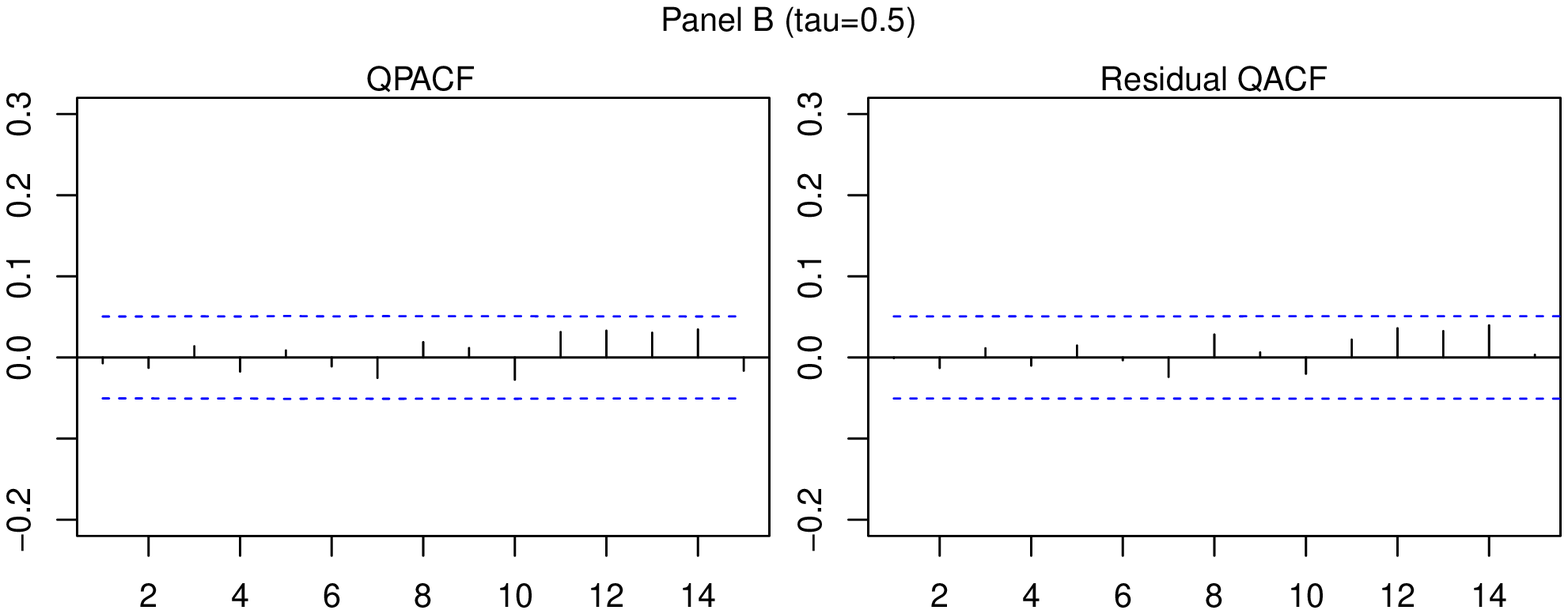}
\includegraphics[scale=0.6]{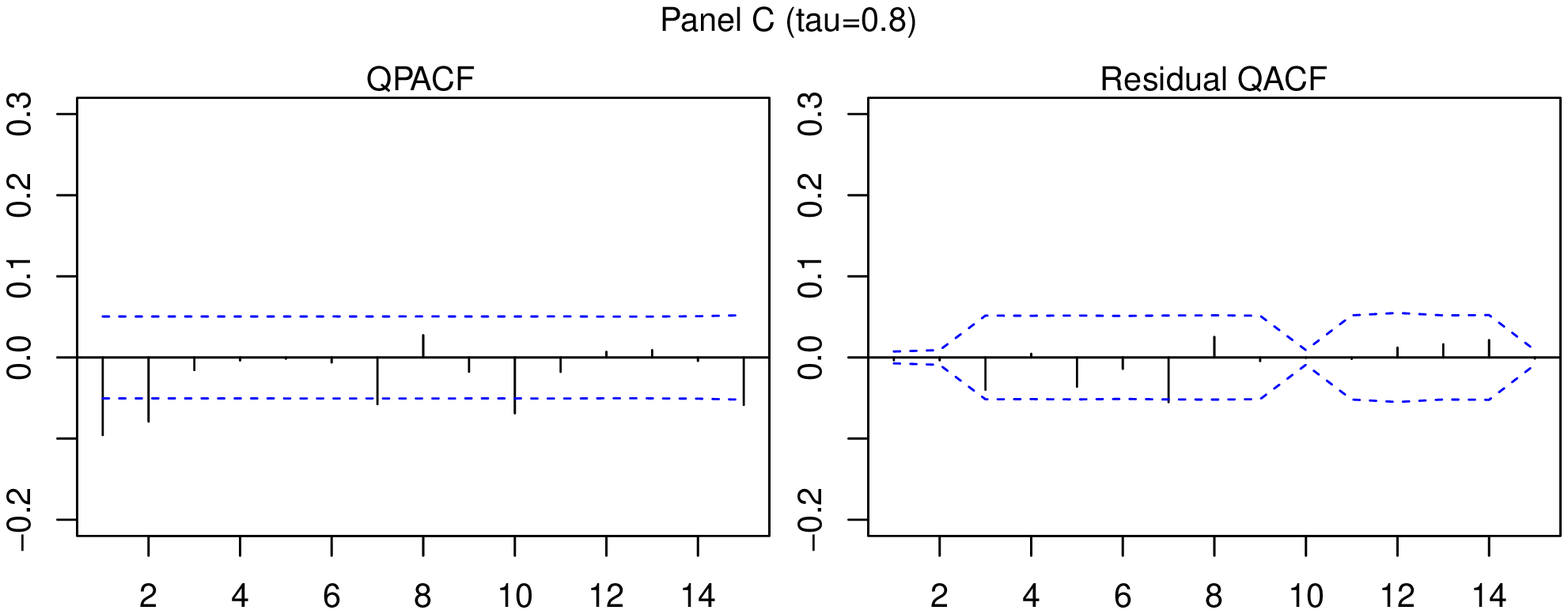}
\caption{\label{nasdaq_fit} The sample QPACF of daily closing prices on the Nasdaq Composite and the sample QACF of residuals from the fitted models for $\tau=0.2$, 0.5, and 0.8. The
dashed lines in the left and right panels correspond to  $\pm 1.96\sqrt{\widehat{\Omega}_3/n}$ and $\pm 1.96\sqrt{\widehat{\Omega}_5/n}$, respectively.}
\end{figure}

\end{document}